\documentclass[11pt, a4paper, english]{article}
\usepackage[utf8]{inputenc}			
\usepackage{makeidx}
\usepackage[pdftex]{graphicx}
\usepackage{mathtools,amsfonts}
\usepackage{array}
\usepackage{multirow, hhline}
\usepackage{paralist}
\usepackage{a4wide}
\usepackage{psfrag}
\usepackage{xspace}

\usepackage[usenames,dvipsnames,table]{xcolor}
\usepackage{wrapfig}
\usepackage[english]{babel}		
\usepackage[margin=1in]{geometry}	
\usepackage{lipsum}			
\usepackage{comment}
\usepackage{chngpage}			
\usepackage{amsfonts,amssymb,amsmath, amsthm}

\usepackage{pgfplots}
\usepackage{pgfgantt}
\usepackage{tikz}
\usetikzlibrary{arrows}
\usetikzlibrary{calc}
\usepackage[noend]{algpseudocode}
\usepackage{algorithmicx}
\usepackage[Algorithm]{algorithm}
\usepackage{subcaption}
\usepackage{tabto}
\usepackage{dirtytalk} 
\usepackage{enumitem}

\usepackage[colorlinks,urlcolor=black,citecolor=black,linkcolor=black,menucolor=black]{hyperref}

\newtheorem{theorem}{Theorem}[section]

\newtheorem{definition}[theorem]{Definition}

\newtheorem{lemma}[theorem]{Lemma}

\newtheorem{fact}[theorem]{Fact}

\newtheorem{observation}[theorem]{Observation}

\newcommand{\R}{\mathbb{R}}

\newcommand{\bigOsymb}{\ensuremath{\mathcal{O}}}

\newcommand{\norm}[2]{\lVert #1 \rVert_{#2}}

\newcommand{\abs}[1]{\lvert #1 \rvert}

\newcommand{\set}[2]{\left\{  #1 \: \middle\vert \: #2 \right\}}
\newcommand{\FS}{\mathit{FS}_{\delta}}
\newcommand{\DP}{\textup{DP}}
\newcommand{\frp}{\mathit{frp}}
\newcommand{\FOV}{\ensuremath{\forall \forall \exists \textup{-OV}}\xspace}
\newcommand{\FOVH}{\forall \forall \exists \mathit{OVH}}

\newcommand{\as}[1]{\mathbf{A}_{#1}}
\newcommand{\bs}[1]{\mathbf{B}_{#1}}
\newcommand{\cs}[1]{\mathbf{C}_{#1}}

\newcommand{\hor}{\mathit{hor}}

\newcommand{\Fr}{Fr{\'e}chet\xspace}
\newcommand{\poly}{\textup{poly}}
\newcommand{\epss}{\varepsilon}

\newcommand{\dO}{\hat{\bigOsymb}}

\newcommand{\dH}{\delta_H}
\newcommand{\dFr}{\delta_F}

\title{Polyline Simplification has Cubic Complexity
}
\date{}
\author{Karl Bringmann\thanks{Max Planck Institute for Informatics, Saarland Informatics Campus} \and Bhaskar Ray Chaudhury\footnotemark[1] \thanks{ Saarbruecken Graduate School of Computer Science}}

\begin{document}

\maketitle

\begin{abstract}
In the classic polyline simplification problem we want to replace a given polygonal curve~$P$, consisting of $n$ vertices, by a subsequence $P'$ of $k$ vertices from $P$ such that the polygonal curves $P$ and $P'$ are as close as possible. Closeness is usually measured using the Hausdorff or Fr\'{e}chet distance. These distance measures can be applied \emph{globally}, i.e., to the whole curves $P$ and $P'$, or \emph{locally}, i.e., to each simplified subcurve and the line segment that it was replaced with separately (and then taking the maximum). This gives rise to four problem variants: Global-Hausdorff (known to be NP-hard), Local-Hausdorff (in time $O(n^3)$), Global-Fr\'{e}chet (in time $O(k n^5)$), and Local-Fr\'{e}chet (in time $O(n^3)$). 

Our contribution is as follows.
\begin{itemize}
\item \emph{Cubic time for all variants:} For Global-Fr\'{e}chet we design an algorithm running in time $O(n^3)$. This shows that all three problems (Local-Hausdorff, Local-Fr\'{e}chet, and Global-Fr\'{e}chet) can be solved in cubic time. All these algorithms work over a general metric space such as $(\mathbb{R}^d,L_p)$, but the hidden constant depends on $p$ and (linearly) on~$d$.

\item \emph{Cubic conditional lower bound:} We provide evidence that in high dimensions cubic time is essentially optimal for all three problems (Local-Hausdorff, Local-Fr\'{e}chet, and Global-Fr\'{e}chet). Specifically, improving the cubic time to $O(n^{3-\epsilon} \poly(d))$ for polyline simplification over $(\mathbb{R}^d,L_p)$ for $p = 1$ would violate plausible conjectures. We obtain similar results for all $p \in [1,\infty), p \ne 2$.
\end{itemize}

In total, in high dimensions and over general $L_p$-norms we resolve the complexity of polyline simplification with respect to Local-Hausdorff, Local-Fr\'{e}chet, and Global-Fr\'{e}chet, by providing new algorithms and conditional lower bounds.
\end{abstract}

\thispagestyle{empty}
\clearpage
\setcounter{page}{1}


\section{Introduction}

We revisit the classic problem of polygonal line simplification, which is fundamental to computational geometry, computer graphics, and geographic information systems. The most frequently implemented and cited algorithms for curve simplification go back to the 70s (Douglas and Peucker~\cite{douglas1973algorithms}) and 80s (Imai and Iri~\cite{ImaiIri1988}).  
These algorithms use the following standard\footnote{The problem was also studied without the restriction that vertices of the simplification belong to the original curve~\cite{Guibas91}. The choice whether the start- and endpoints of $P$ and $P'$ must coincide or not is typically irrelevant in this area; in this paper we assert that they coincide, but all results could also be proved without this assumption. Moreover, sometimes $\delta$ is given and $k$ should be minimized, sometimes $k$ is given and $\delta$ should be minimized; we focus on the former variant in this paper.} formalization of curve simplification.
A \emph{polygonal curve} or \emph{polyline} is given by a sequence $P = \langle v_0,v_1,\ldots,v_n \rangle$ of points $v_i \in \R^d$, and represents the continuous curve walking along the line sequences $\overline{v_i v_{i+1}}$ in order.
Given a polyline $P = \langle v_0,v_1,\ldots,v_n \rangle$ and a number $\delta > 0$, we want to compute a subsequence $P' = \langle P_{i_0},\ldots,P_{i_{k-1}}\rangle$, with $0 = i_0 < \ldots < i_{k-1} = n$, of minimal length $k$ such that $P$ and $P'$ have ``distance'' at most~$\delta$. 

Several distance measures have been used for the curve simplification problem. The most generic distance measure on \emph{point sets} $A,B$ is the \emph{Hausdorff distance}. The (directed) Hausdorff distance from $A$ to $B$ is the maximum over all $a \in A$ of the distance from $a$ to its closest point in $B$. 
This is used on curves $P,Q$ by applying it to the images of the curves in the ambient space, i.e., to the union of all line segments $\overline{v_i v_{i+1}}$.

However, the most popular distance measure for curves in computational geometry is the \emph{\Fr distance} $\dFr$. This is the minimal length of a leash connecting a dog to its owner as they continuously walk along the two polylines without backtracking. In comparison to Hausdorff distance, it takes the ordering of the vertices along the curves into account, and thus better captures an intuitive notion of distance among curves.

For both of these distance measures $\delta_{*} \in \{\dH, \dFr\}$, we can apply them \emph{locally} or \emph{globally} in order to measure the distance between the original curve $P$ and its simplification $P'$. In the global variant, we simply consider the distance $\delta_{*}(P,P')$, i.e., we use the Hausdorff or \Fr distance of $P$ and $P'$. In the local variant, we consider the distance $\max_{1 \le \ell< k} \delta_{*}(P[i_{\ell-1} \ldots i_{\ell}], \overline{v_{i_{\ell-1}} v_{i_{\ell}}})$, i.e., for each simplified subcurve $P[i_{\ell-1} \ldots i_{\ell}]$ of $P$ we compute the distance to the line segment $\overline{v_{i_{\ell-1}} v_{i_{\ell}}}$ that we simplified the subcurve to, and we take the maximum over these distances.
This gives rise to four problem variants, depending on the distance measure: Local-Hausdorff, Local-\Fr, Global-Hausdorff, and Global-\Fr. See Section~\ref{preliminaries} for details.

Among these variants, Global-Hausdorff is unreasonable in that it essentially does not take the ordering of vertices along the curve into account. Moreover, it was recently shown that curve simplification under Global-Hausdorff is NP-hard~\cite{kreveld2018}. For these reasons, we do not consider this measure in this paper.

The classic algorithm by Imai and Iri~\cite{ImaiIri1988} was designed for Local-Hausdorff simplification and solves this problem in time\footnote{In $\dO$-notation we \emph{hide any polynomial factors in $d$}, but we make exponential factors in $d$ explicit.} $\dO(n^3)$. By exchanging the distance computation in this algorithm for the \Fr distance, one can obtain an $\dO(n^3)$-time algorithm for Local-\Fr~\cite{Godausimplification15}. Several papers obtained improvements for Local-Hausdorff simplification in small dimension~$d$ \cite{melkman1988, chanchin1996, Barequet2002}; the fastest known running times are $2^{O(d)} n^2$ for $L_1$-norm, $\dO(n^2)$ for $L_\infty$-norm, and $\dO(n^{3 - \Omega(1/d)})$ for $L_2$-norm~\cite{Barequet2002}.

The remaining variant, Global-\Fr, has only been studied very recently~\cite{kreveld2018}, although it is a reasonable measure: The Local constraints (i.e., matching each $v_{i_\ell}$ to itself) are not necessary to enforce ordering along the curve, since \Fr distance already takes the ordering of the vertices into account -- in contrast to Hausdorff distance, for which the Local constraints are necessary to enforce any ordering.
More generally, Global-\Fr simplification is very well motivated as \Fr distance is a popular distance measure in computational geometry, and Global-\Fr simplification exactly formalizes curve simplification with respect to the \Fr distance. 
Van Kreveld et al.~\cite{kreveld2018} presented an algorithm for Global-\Fr simplification in time $\dO(k^* \cdot n^5)$, where $k^*$ is the output size, i.e., the size of the optimal simplification. 

\subsection{Contribution 1: Algorithm for Global-\Fr}

From the state of the art, one could get the impression that Global-\Fr is a well-motivated, but computationally expensive curve simplification problem, in comparison to Local-Hausdorff and Local-\Fr.
We show that the latter intuition is wrong, by designing an $\dO(n^3)$-time algorithm for Global-\Fr simplification. This is an improvement by a factor $\Theta(k^* \cdot n^2)$ over the previously best algorithm by van Kreveld et al.~\cite{kreveld2018}.

\begin{theorem}[Section~\ref{sec:algo}] \label{mainAlgo}
  Global-\Fr simplification can be solved in time $\dO(n^3)$.
\end{theorem}

This shows that all three problem variants (Local-Hausdorff, Local-\Fr, and Global-\Fr) can be solved in time $\dO(n^3)$, and thus the choice of which problem variant to apply should not be made for computational reasons, at least in high dimensions.

Our algorithm (as well as the algorithms for Local-Hausdorff and Local-\Fr~\cite{ImaiIri1988, Godausimplification15}) works over a general metric space such as $\mathbb{R}^d$ with $L_p$-norm. The hidden constant depends on $p$, and has linear dependence on $d$ (recall that in $\dO$-notation we hide polynomial factors in $d$). We assume the Real RAM model of computation, which allows us to perform exact distance computations, and to exactly solve equations of the form $\norm{x - a}{p} = r$ for given $a \in \mathbb{R}^d$, $r \ge 0$. See Section~\ref{sec:techoverview} for an overview of the algorithm.

\subsection{Contribution 2: Conditional Lower Bound}

Since all three variants can be solved in time $\dO(n^3)$, the question arises whether any of them can be solved in time $\dO(n^{3-\epss})$. Tools to (conditionally) rule out such algorithms have been developed in recent years in the area of \emph{fine-grained complexity}, see, e.g., the survey~\cite{vwilliamsSurvey}. One of the most widely used fine-grained hypotheses is the following.

\medskip \noindent
\textbf{\boldmath$k$-OV Hypothesis:} \emph{Problem:} Given sets $A_1,\ldots,A_k \subseteq \{0,1\}^d$ of size $n$, determine whether there exist vectors $a_1 \in A_1,\ldots,a_k \in A_k$ that are orthogonal, i.e., for each dimension $j \in [d]$ there is a vector $i \in [k]$ with $a_i[j] = 0$. \\ \emph{Hypothesis:} For any $k \ge 2$ and $\varepsilon > 0$ the problem is not in time $\dO(n^{k-\varepsilon})$.
\medskip

Naively, $k$-OV can be solved in time $\dO(n^k)$, and the hypothesis asserts that no polynomial improvement is possible, at least not with polynomial dependence on $d$. See~\cite{AmirWilliams15} for the fastest known algorithms for $k$-OV.

Buchin et al.~\cite{buchin2016fine} used the 2-OV hypothesis to rule out $\dO(n^{2-\varepsilon})$-time algorithms for Local-Hausdorff\footnote{Their proof can be adapted to also work for Local-\Fr and Global-\Fr.} in the $L_1$, $L_2$, and $L_\infty$ norm. This yields a tight bound for $L_\infty$, since an $\dO(n^2)$-time algorithm is known~\cite{Barequet2002}. 
However, for all other $L_p$-norms ($p \in [1,\infty)$), the question remained open whether $\dO(n^{3-\epss})$-time algorithms exist. 
To answer this question, one could try to generalize the conditional lower bound by Buchin et al.~\cite{buchin2016fine} to start from 3-OV. However, curve simplification problems seem to have the wrong ``quantifier structure'' for such a reduction. See Section~\ref{sec:techoverview} below for more intuition. 
For similar reasons, Abboud et al.~\cite{Amirhittingset16} introduced the Hitting Set hypothesis, in which they essentially consider a variant of 2-OV where we have a universal quantifier over the first set of vectors and an existential quantifier over the second one ($\forall \exists$-OV). From their hypothesis, however, it is not known how to prove higher lower bounds than quadratic. We therefore consider the following natural extension of their hypothesis. This problem was studied in a more general context by Gao et al.~\cite{GaoImpagliazzo17}.

\medskip \noindent
\textbf{\boldmath$\forall \forall \exists$-OV Hypothesis:} \emph{Problem:} Given sets $A, B, C \subseteq \{0,1\}^d$ of size $n$, determine whether for all $a \in A, b \in B$ there exists $c \in C$ such that $a,b,c,$ are orthogonal. \\ \emph{Hypothesis:} For any $\varepsilon > 0$ the problem is not in time $\dO(n^{3-\varepsilon})$.
\medskip

No algorithm violating this hypothesis is known, and even for much stronger hypotheses on variants of $k$-OV and Satisfiability no such algorithms are known, see Section~\ref{sec:discussion} for details. This shows that the hypothesis is plausible, in addition to being a natural generalization of the hypothesis of Abboud et al.~\cite{Amirhittingset16}.

We establish a \FOV-based lower bound for curve simplification.

\begin{theorem}[Section~\ref{sec:lowerbound}] \label{mainLower}
  Over $(\mathbb{R}^d,L_p)$ for any $p \in [1,\infty)$ with $p \ne 2$, Local-Hausdorff, Local-\Fr, and Global-\Fr simplification have no $\dO(n^{3-\epss})$-time algorithm for any $\epss > 0$, unless the \FOV hypothesis fails.
\end{theorem}

In particular, this rules out improving the $2^{O(d)} n^2$-time algorithm for Local-Hausdorff over $L_1$~\cite{Barequet2002} to a polynomial dependence on $d$. Note that the theorem statement excludes two interesting values for $p$, namely $\infty$ and 2. For $p = \infty$, an $\dO(n^2)$-time algorithm is known for Local-Hausdorff~\cite{Barequet2002}, so proving the above theorem also for $p = \infty$ would immediately yield an algorithm breaking the \FOV hypothesis. For $p=2$, we do not have such a strong reason why it is excluded, however, we argue in Section~\ref{sec:techoverview} that at least a significantly different proof would be necessary in this case. This leaves open the possibility of a faster curve simplification algorithm for $L_2$, but such a result would need to exploit the Euclidean norm very heavily.

\subsection{Technical Overview} \label{sec:techoverview}

\paragraph{Algorithm} 
We first sketch the algorithm by Imai and Iri~\cite{ImaiIri1988} for Local-Hausdorff. Given a polyline $P = \langle v_0,\ldots,v_n \rangle$ and a distance threshold $\delta$, for all $i < i'$ we compute the Hausdorff distance $\delta_{i,i'}$ from the subpolyline $P[i\ldots i']$ to the line segment $\overline{v_i v_{i'}}$. This takes total time $O(n^3)$, since Hausdorff distance between a polyline and a line segment can be computed in linear time. We build a directed graph on vertices $\{0,1,\ldots,n\}$, with a directed edge from $i$ to $i'$ if and only if $\delta_{i,i'} \le \delta$. We then determine the shortest path from 0 to $n$ in this graph. This yields the simplification $P'$ of smallest size, with Local-Hausdorff distance at most $\delta$. The running time is dominated by the first step, and is thus $O(n^3)$. Replacing Hausdorff by \Fr distance yields an $O(n^3)$-time algorithm for Local-\Fr. 

Note that these algorithms are simple dynamic programming solutions. For Global-\Fr, our cubic time algorithm also uses dynamic programming, but is significantly more complicated.

In our algorithm, we compute the same dynamic programming table as the previously best algorithm~\cite{kreveld2018}. This is a table of size $O(k^* \cdot n^2)$, where $k^*$ is the output size. Table entry $\DP(k,i,j)$ stores the earliest reachable point on the line segment $\overline{v_j v_{j+1}}$ with a size-$k$ simplification of $P[0\ldots i]$. 
More precisely, $\DP(k,i,j)$ is the minimal $t$, with $j \le t \le j+1$, such that there is a size-$k$ simplification $P'$ of $P[0\ldots i]$ with $\delta_F(P', P[0\ldots t]) \le \delta$. If such a point does not exist, we set $\DP(k,i,j) = \infty$.

A simple algorithm computes a table entry in time $O(n^3)$: We iterate over all possible second-to-last points $v_{i'}$ of the simplification $P'$, and over all possible previous line segments $\overline{v_{j'} v_{j'+1}}$, and check whether from $i'$ on $P'$ and $\DP(k-1,i',j')$ on $P$ we can ``walk'' to $i$ on $P'$ and some $j \le t \le j+1$ and $P$, always staying within the required distance. Moreover, we compute the earliest such $t$. This can be done in time $O(n^3)$, which in total yields time $O(k^* n^5)$. This is the algorithm from~\cite{kreveld2018}. 

In order to obtain a speedup, we split the above procedure into two types: $j' = j$, i.e., the walks ``coming from the left'', and $j' < j$, i.e., the walk ``coming from the bottom''. For the first type, it can be seen that the simple algorithm computes their contribution to the output in time $O(n)$. Moreover, it is easy to bring down this running time to $O(1)$ per table entry, by maintaining a certain minimum. 

We show how to handle the second type in total time $O(n^3)$. This is the bulk of effort going into our new algorithm. Here, the main observation is that the particular values of $\DP(k-1,i',j')$ are irrelevant, and in particular we only need to store for each $i',j'$ the smallest $k'$ such that $\DP(k',i',j') \ne \infty$. 
Using this observation, and further massaging the problem, we arrive at the following subproblem that we call \emph{Cell Reachability}: 
We are given $n$ squares (or \emph{cells}) numbered $1,\ldots,n$ and stacked on top of each other. Between cell $j$ and cell $j+1$ there is a \emph{passage}, which is an interval on their common boundary through which we can pass from $j$ to $j+1$. Finally, we are given an integral \emph{entry-cost} $\lambda_j$ for each cell $j$. The goal is to compute, for each cell $j$, its \emph{exit-cost} $\mu_j$, defined as the minimal entry-cost $\lambda_{j'}$, $j' < j$, such that we can walk from cell $j'$ to cell $j$ through the contiguous passages in a monotone fashion (i.e., the points at which we cross a passage are monotonically non-decreasing). See Figure~\ref{Cellreachability1} for an illustration of this problem.

To solve Cell Reachability, we determine for each cell $j$ and cost $k$ the leftmost point $t_j(k)$ on the passage from cell $j-1$ to cell $j$ at which we can arrive from some cell $j' < j$ with entry-cost at most $k$ (using a monotone path). Among the sequence $t_j(1),t_j(2),\ldots$ we only need to store the break-points, with $t_j(k) < t_j(k-1)$, and we design an algorithm to maintain these break-points in amortized time $O(1)$ per cell $j$. This yields an $O(n)$-time solution to Cell Reachability, which translates to an $O(n^3)$-time solution to Global-\Fr simplification.

\paragraph{Conditional lower bound}
Let us first briefly sketch the previous conditional lower bound by Buchin et al.~\cite{buchin2016fine}. Given a 2-OV instance on vectors $A,B \subseteq \{0,1\}^d$, they construct corresponding point sets $\tilde A, \tilde B \subset \mathbb{R}^{d'}$ (for some $d' = O(d)$), forming two clusters that are very far apart from each other. They also add a start- and an endpoint, which can be chosen far away from these clusters (in a new direction). Near the midpoint between $\tilde A$ and $\tilde B$, another set of points $\tilde C$ is constructed. The final curve then starts in the startpoint, walks through all points in $\tilde A$, then through all points in $\tilde C$, then through all points in $\tilde B$, and ends in the endpoint.
This setup ensures that any reasonable size-4 simplification must consist of the startpoint, one point $\tilde a \in \tilde A$, one point $\tilde b \in \tilde B$, and the endpoint. All points in $\tilde A$ are close to $\tilde a$, so they are immediately close to the simplification, similarly for $\tilde B$. Thus, the constraints are in the points $\tilde C$. Buchin et al.~\cite{buchin2016fine} construct $\tilde C$ such that it contains one point for each dimension $\ell \in [d]$, which ``checks'' that the vectors corresponding to the chosen points $\tilde a, \tilde b$ are orthogonal in dimension $\ell$, i.e., one of $a$ or $b$ has a 0 in dimension $\ell$.

We instead want to reduce from \FOV, so we are given an instance $A,B,C$ and want to know whether for all $a \in A, b \in B$ there exists $c \in C$ such that $a,b,c$ are orthogonal. 
In our adapted setup, the set $\tilde C$ is in one-to-one correspondence to the set of vectors $C$. That is, choosing a size-4 simplification implements an existential quantifier over $a \in A, b \in B$, and the contraints that all $\tilde c \in \tilde C$ are close to the line segment from $\tilde a$ to $\tilde b$ implements a universal quantifier over $c \in C$. Naturally, we want the distance from $\tilde c$ to the line segment $\overline{\tilde a \tilde b}$ to be large if $a,b,c$ are orthogonal, and to be small otherwise. This simulates the negation of \FOV, so any curve simplification algorithm can be turned into an algorithm for \FOV.

The restriction $p \in [1,\infty)$ with $p \ne 2$ in Theorem~\ref{mainLower} already is a hint that the specific construction of points is subtle. Indeed, let us sketch one critical issue in the following. We want the points $\tilde C$ to lie in the middle between $\tilde A$ and $\tilde B$, which essentially means that we want to consider the distance from $(\tilde a+\tilde b)/2$ to $\tilde c$. Now consider just a single dimension. Then our task boils down to constructing points $a_0,a_1$ and $b_0,b_1$ and $c_0,c_1$, corresponding to the bits in this dimension, such that $\norm{(a_i+b_j)/2 - c_k}{p} = \beta_1$ if $i=j=k=1$ and $\beta_0$ otherwise, with $\beta_1 < \beta_0$. 
Writing $a'_i = a_i/2$ and $b'_j = b_j/2$ for simplicity, in the case $p=2$ we can simplify 
\begin{align}
&\norm{a'_i+b'_j - c_k}{2}^2
= \sum_{\ell=1}^{d'} (a'_i[\ell] + b'_j[\ell] - c_k[\ell])^2 \notag \\
&= \sum_{\ell=1}^{d'} \Big( (a'_i[\ell] + b'_j[\ell])^2 + (a'_i[\ell] - c_k[\ell])^2 + (b'_j[\ell] - c_k[\ell])^2  - a'_i[\ell]^2 - b'_j[\ell]^2 - c_k[\ell]^2 \Big) \notag \\
&= \norm{a'_i + b'_j}{2}^2 + \norm{a'_i - c_k}{2}^2 + \norm{b'_j - c_k}{2}^2 - \norm{a'_i}{2}^2 - \norm{b'_j}{2}^2 - \norm{c_k}{2}^2 \notag  \\
&= f_1(i,j) + f_2(j,k) + f_3(i,k), \label{eq:fiRepresentation}
\end{align}
for some functions\footnote{This holds for $f_1(i,j) := \norm{a'_i + b'_j}{2}^2 - \norm{a'_i}{2}^2$, $f_2(j,k) := \norm{b'_j - c_k}{2}^2 - \norm{b'_j}{2}^2$, and $f_3(i,k) := \norm{a'_i - c_k}{2}^2 - \norm{c_k}{2}^2$.} $f_1,f_2,f_3 \colon \{0,1\} \times \{0,1\} \to \mathbb{R}$. Note that by assumption this is equal to $\beta_1^2$ if $i=j=k=1$ and $\beta_0^2$ otherwise, with $\beta_1 < \beta_0$. After a linear transformation, we thus obtain a representation of the form (\ref{eq:fiRepresentation}) for the function $f(i,j,k) = i \cdot j \cdot k$ for $i,j,k \in \{0,1\}$.
However, it can be checked that such a representation is impossible\footnote{For instance, we can express this situation by a linear system of equations in 12 variables (the 4 image values for each function $f_i$) and 8 equations (for the values of $f$ on $i,j,k \in \{0,1\}$) and verify that it has no solution.}. Therefore, for $p=2$ \emph{our outlined reduction cannot work - provably!} 

We nevertheless make this reduction work in the cases $p \in [1,\infty)$, $p \ne 2$. The above argument shows that the construction is necessarily subtle. Indeed, constructing the right points requires some technical effort, see Section~\ref{sec:lowerbound}.

\subsection{Further Related Work}
Curve simplification has been studied in a variety of different formulations and settings, and it is well beyond the scope of this paper to give an overview. To list some examples, it was shown that the classic heuristic algorithm by Douglas and Peucker~\cite{douglas1973algorithms} can be implemented in time $O(n \log n)$~\cite{Hershberger1994}, and that the classic $O(n^3)$-time algorithm for Local-Hausdorff simplification by Imai and Iri~\cite{ImaiIri1988} can be implemented in time $O(n^2)$ in two dimensions~\cite{chanchin1996,melkman1988}. Further topics include curve simplification without self-intersections~\cite{Schirraberg1998}, Local-Hausdorff simplification with additional constraints on the angles between consecutive line segments~\cite{chenhershberger2005}, approximation algorithms~\cite{agarwal2005near}, streaming algorithms~\cite{Abam2010}, and the use of curve simplification in subdivision algorithms~\cite{Guibas91,estkowski2001simplifying,funke17}.

\subsection{Organization}

In Section~\ref{preliminaries} we formally define the problems studied in this paper. In Section~\ref{sec:algo} we present our new algorithm for Global-\Fr simplification, and in Section~\ref{sec:lowerbound} we show our conditional lower bounds. We further discuss the used hypothesis in Section~\ref{sec:discussion}.

\section{Preliminaries} \label{preliminaries}
Our ambient space is the metric space $(\mathbb{R}^d,L_p)$, where the distance between points $x,y \in \mathbb{R}^d$ is the $L_p$-norm of their difference, i.e., $\norm{x - y}{p} = \big(\sum_{i=1}^d (x[i] - y[i])^p \big)^{1/p}$. 
A \emph{polyline} $P$ of size $n$ is given by a sequence of points $\langle v_0,v_1,\ldots ,v_n \rangle$, where each $v_i$ lies in the ambient space. We associate with $P$ the continuous curve that starts in $v_0$, walks along the line segments $\overline{v_i v_{i+1}}$ for $i=0,\ldots,n-1$ in order, and ends in $v_n$.
We also interpret $P$ as a function $P \colon [0,n] \rightarrow \mathbb{R}^d$ where $P[i + \lambda] = (1-\lambda)v_i + \lambda v_{i+1}$ for any $\lambda \in [0,1]$ and $i \in \{0,\ldots,n-1\}$. 
We use the notation $P[t_1 \ldots t_2]$ to represent the sub-polyline of $P$ between $P[t_1]$ and $P[t_2]$. Formally for any integers $0 \leq i \leq j \leq n$ and reals $\lambda_1 \in [0,1)$ and $\lambda_2 \in (0,1]$, 
\begin{align*}
  P[i+ \lambda_1 \ldots j+\lambda_2] &= \langle (1-\lambda_1)v_i + \lambda_1v_{i+1}, v_{i+1}, \ldots , v_j, (1- \lambda_2)v_j + \lambda_2 v_{j+1} \rangle
\end{align*}

A \emph{simplification} of $P$ is a curve $Q = \langle v_{i_0},v_{i_1} \ldots ,v_{i_m} \rangle$ with $0 = i_0 < i_1 < \ldots < i_m = n$. The size of the simplification $Q$ is $m+1$. Our goal is to determine a simplification of given size $k$ that \say{very closely} represents $P$. To this end we define two popular measures of similarity between the curves, namely the Fr\'echet and Hausdorff distances. 

\begin{definition}[Fr\'echet distance]
 \label{frechetdefinition}
  The (continuous) Fr\'echet distance $\delta_F(P_1,P_2)$ between two curves $P_1$ and $P_2$ of size $n$ and $m$ respectively is 
   \begin{align*}
     \delta_F(P_1,P_2) &= \inf_{f} \max \limits_{t \in [0,n]} \norm{P_1[t] - P_2[f(t)]}{p}
   \end{align*} 
  where $f \colon [0,n] \rightarrow [0,m]$ is monotone with $f(0) = 0$ and $f(n) = m$.
\end{definition} 

Alt and Godau~\cite{AltG95} gave the characterization of the Fr\'echet distance in terms of the so-called free-space diagram.

\begin{definition}[Free-Space]
 Given two curves $P_1$, $P_2$ and $\delta \geq  0$, the free-space $\FS(P_1,P_2) \subseteq \mathbb{R}^2$ is the set $\set{(x,y) \in ([0,n] \times [0,m])}{\norm{P_1[x]-P_2[y]}{p} \leq \delta}$.
\end{definition}

Consider the following decision problem. Given two curves $P_1$, $P_2$ of size $n$ and $m$, respectively, and given $\delta \geq 0$, decide whether $\delta_F(P_1,P_2) \leq \delta$. The answer to this question is yes if and only if $(n,m)$ is reachable from $(0,0)$ by a monotone path through $\FS(P_1,P_2)$. This \say{reachability} problem is known to be solvable by a dynamic programming algorithm in time $\mathcal{O}(nm)$, and the standard algorithm for computing the Fr\'echet distance is an adaptation of this decision algorithm~\cite{AltG95}. In particular, if either $P_1$ or $P_2$ is a line segment, then the decision problem can be solved in linear time. 

 The Hausdorff distance between curves ignores the ordering of the points along the curve. Intuitively, if we remove the monotonicity condition from function $f$ in Definition \ref{frechetdefinition} we obtain the directed Hausdorff distance between the curves. Formally, it is defined as follows.

\begin{definition}[Hausdorff distance]
 \label{hausdorffdefinition}
   The (directed) \emph{ Hausdorff} distance $\delta_{H}(P_1,P_2)$ between curves $P_1$ and $P_2$ of size $n$ and $m$, respectively, is 
   \begin{align*}
     \delta_{H}(P_1,P_2) &= \max \limits_{t_1 \in [0,n]} \min \limits_{t_2 \in [0,m]} \norm{P_1[t_1] - P_2[t_2]}{p}.
   \end{align*}
\end{definition} 
 
In order to measure the \say{closeness} between a curve and its simplification, these above similarity measures can be applied either \emph{globally} to the whole curve and its simplification, or \emph{locally} to each simplified subcurve $P[i_\ell \ldots i_{\ell+1}]$ and the segment $\overline{v_{i_{\ell}},v_{i_{\ell+1}}}$ to which it was simplified (taking the maximum over all $\ell$). This gives rise to the following measures for curve simplification.
\begin{definition}[Similarity for Curve Simplification]
  \label{local-distances}
   Given a curve $P = \langle v_0,v_1,\ldots ,v_n \rangle$ and a simplification $Q = \langle v_{i_0},v_{i_1} \ldots ,v_{i_m} \rangle$ of $P$, we define their
   \begin{itemize}
     \item Global-Hausdorff distance as $\delta_H(P,Q)$,
     \item Global-Fr\'echet distance as $\delta_F(P,Q)$,  
     \item Local-Hausdorff distance\footnote{It can be checked that in this expression directed and undirected Hausdorff distance have the same value, and so for Local-Hausdorff we can without loss of generality use the directed Hausdorff distance. For Global-Hausdorff this choice makes a difference, but we do not consider this problem in this paper.} as $\max \limits_{0 \leq \ell \leq m-1} \delta_H(P[i_{\ell} \ldots i_{\ell+1}], \overline{v_{i_\ell}v_{i_{\ell+1}}})$, and
     \item Local-Fr\'echet distance as $\max \limits_{0 \leq \ell \leq m-1} \delta_F(P[i_{\ell} \ldots i_{\ell+1}], \overline{v_{i_\ell}v_{i_{\ell+1}}})$.   
   \end{itemize}
\end{definition}

\section{Algorithms for Global-Fr\'echet simplification} \label{sec:algo}
In this section we present an $\mathcal{O}(n^3)$ time algorithm for curve simplification under Global-Fr\'echet distance, i.e., we prove Theorem~\ref{mainAlgo}.
\subsection{An $O(kn^5)$ algorithm for Global Fr\'echet simplification}
We start by describing the previously best algorithm by ~\cite{kreveld2018}. Let $P$ be the polyline $\langle v_0,v_1, \ldots v_n \rangle$. Let $\DP(k,i,j)$ represent the earliest reachable point on $\overline{v_jv_{j+1}}$ with a length $k$ simplification of the polyline $P[0 \ldots i]$ i.e.  $\DP(k,i,j)$ represents the smallest $t$ such that $P[t]$ lies on the line-segment $\overline{v_{j}v_{j+1}}$ (i.e. $j \leq t \leq j+1$) and there is a simplification $\tilde{Q}$ of the polyline $P[0 \ldots i]$ of size at most $k$ such that $\delta_F(\tilde{Q},P[0 \ldots t]) \leq \delta$. If such a point does not exist then we set $\DP(k,i,j) = \infty$. To solve Global-Fr\'echet simplification,  we need to return the minimum $k$ such that $\DP(k,n,n-1) \neq \infty$. Let $P[t_{i,j}]$ and $P[s_{i,j}]$ be the first point and the last point respectively on the line segment $\overline{v_jv_{j+1}}$ such that $\norm{v_i -P[t_{i,j}]}{p} \leq \delta$ and $\norm{v_i - P[s_{i,j}]}{p} \leq \delta$. Observe that if $\DP(k,i,j) \neq \infty$ then $t_{i,j} \leq \DP(k,i,j) \leq s_{i,j}$ for all $k$. 
Before moving onto the algorithm we make some simple observations,

\begin{observation}
 \label{dp-property1}
  If $\DP(k,i,j) = \infty$ then $\DP(k',i,j) = \infty$ for all $k'< k$. If $\DP(k,i,j) = t_{i,j}$ then $\DP(k',i,j) = t_{i,j}$ for all $k'\geq k$.
\end{observation}

\begin{proof}
  If $k' < k$, then the minimization in $\DP(k,i,j)$ is over a superset compared to $\DP(k',i,j)$. Thus $\DP(k',i,j) \geq \DP(k,i,j) = \infty$. Thus $\DP(k,i,j) = \infty$. Similarly when $k' \geq k$ then the minimization in $\DP(k',i,j)$ is over a superset compared to $\DP(k,i,j)$. Thus we have $t_{i,j} \leq \DP(k',i,j) \leq \DP(k,i,j)$. Thus $\DP(k,i,j) = t_{i,j}$ implies $\DP(k',i,j) = t_{i,j}$ for all $k' \geq k$. 
\end{proof}
We will crucially make use of the following characterization of the $\DP$ table entries,
\begin{lemma}
\label{recurrence}
   $\DP(k,i,j)$ is the minimal $t \in [t_{i,j},s_{i,j}]$, such that for some $i' < i$ and $j' \leq j$, we have $\DP(k-1,i',j') \neq \infty$ and $\delta_F(P[\DP(k-1,i',j') \ldots t], \overline{v_{i'}v_{i}}) \leq \delta$. If no such $t$ exists then $\DP(k,i,j) = \infty$.
\end{lemma}

  \begin{proof}
  
   \begin{figure}
    \begin{subfigure}{0.5\textwidth}
     \includegraphics[scale = 0.5]{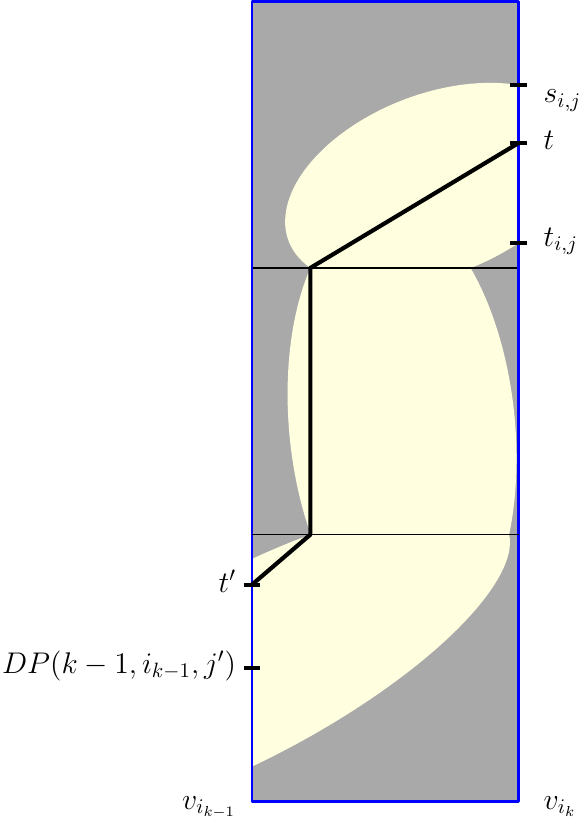}
    \end{subfigure}
     \begin{subfigure}{0.5\textwidth}
     \includegraphics[scale = 0.5]{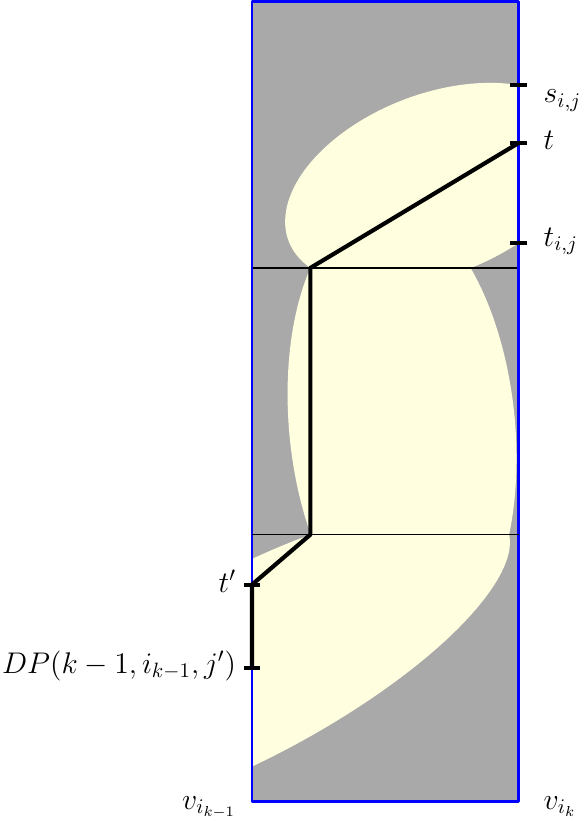}
  \end{subfigure}
   \caption{Illustration of the proof of Lemma \ref{recurrence}. There exists a monotone path from $(0,t')$ to $(1,t)$ in $\FS(P,\overline{v_{i_{k-1}}v_{i_k}})$ (left). Since $\DP(k-1,i_{k-1},j') \leq t' \leq t$, there is a monotone path in $\FS(P, \overline{v_{i_{k-1}}v_{i_k}})$ from $(0,\DP(k-1,i_{k-1},j'))$ (right) to $(1,t)$ by moving from $(0,\DP(k-1,i_{k-1},j'))$ to $(0,t')$ and then following the existing monotone path from $(0,t')$ to $(1,t)$.}
   \label{reachability}
\end{figure}

   Let $t$ be minimal in $[t_{i,j},s_{i,j}]$ such that $\DP(k-1,i',j') \neq \infty$ and $\delta_F(P[\DP(k-1,i',j') \ldots t],$ $\overline{v_{i'}v_{i}}) \leq \delta$ for some $i' < i$ and $j' \leq j$. Since in particular $\DP(k-1,i',j') \neq \infty$, for one direction we note that there exists a simplification $\hat{Q}$ of the polyline $P[0 \ldots i']$  of size $k-1$ such that $\delta_F(\hat{Q},P[ 0 \ldots \DP(k-1,i',j')]) \leq \delta$. By appending $v_j$ to $\hat{Q}$ we obtain a simplification $\tilde{Q}$ of the polyline $P[0 \ldots i]$ such that $\delta_F(\tilde{Q},P[0 \ldots \hat{t}]) \leq \max(\delta_F(\hat{Q},P[0 \ldots \DP(k-1,i',j')]),\delta_F(P[\DP(k-1,i',j') \ldots t], \overline{v_{i'}v_{i}})) \leq \delta$. It follows that $\DP(k,i,j) \leq t$. In particular if $\DP(k,i,j) = \infty$ then no such $t$ exists.
   
   For the other direction, let $t'$ be such that $\DP(k,i,j) = t'$. Assume $t' \neq \infty$. Then there exists a simplification $\tilde{Q} = \langle v_{i_0}, v_{i_1}, \ldots v_{i_{k-1}},v_{i_k} \rangle$ of the Polyline $P[0 \ldots i]$ such that $\delta_F(\tilde{Q},P[0 \ldots t']) \leq \infty$. Such a $\tilde{Q}$ exists if and only if there is a simplification $\hat{Q}$ of size $k-1$ of the polyline $\hat{P} = P[0 \ldots i_{k-1}]$ and a value $\hat{t} \leq t'$ such that, 
   
   \begin{enumerate}[label={(\arabic*)}]
       \item $\delta_{F}(\hat{Q},P[0 \ldots \hat{t}]) \leq \delta$ and 
       \item $\delta_{F}(P[\hat{t} \ldots t'], \overline{v_{i_{k-1}}v_{i_k}}) \leq \delta$.
   \end{enumerate}
   
    Let $\hat{j} \leq \hat{t} \leq \hat{j}+1$. Observe that (1) implies that $\DP(k-1,i_{k-1},\hat{j}) \neq \infty$.  Also $t_{i_{k-1},\hat{j}} \leq  \DP(k-1,i_{k-1},\hat{j}) \leq \hat{t} \leq s_{i_{k-1},\hat{j}}$. Now we show that $\delta_{F}(P[\hat{t}\ldots t'], \overline{v_{i_{k-1}}v_{i_k}}) \leq \delta$ implies that $\delta_{F}(P[ \DP(k-1,i_{k-1},\hat{j})\ldots t'], \overline{v_{i_{k-1}}v_{i_k}}) \leq \delta$. This is obvious from inspecting $\FS(P, \overline{v_{i_{k-1}}v_{i_k}})$ (see Figure \ref{reachability}). There exists a monotone path in $\FS(P,\overline{v_{i_{k-1}}v_{i_k}})$ that starts from  $(0,\DP(k-1,i_{k-1},\hat{j}))$ , moves to $(0,\hat{t})$ and then follows the monotone path from $(0,\hat{t})$ to $(1,t')$ that exists. Therefore $t \leq t' = \DP(k,i,j)$. Combining the two inequalities we have that $\DP(k,i,j) = t$.
\end{proof}

A dynamic programming algorithm follows more or less directly from Lemma \ref{recurrence}. Note that for a fixed $i' <i$ and $j' \leq j$ such that $\DP(k-1,i',j') \neq \infty$ we can determine the minimal $t$ such that $(1,t)$ is reachable from $(0,\DP(k-1,i',j'))$ by a monotone path in $\FS(P,\overline{v_{i'}v_{i}})$ in $\mathcal{O}(n)$ time. This follows from the standard algorithm for the decision version of the Fr\'echet distance  between two polygonal curves of length at most $n$ (in particular here one of the curves is of length 1). To determine $\DP(k,i,j)$ we enumerate over all $i' < i$ and $j' \leq j$ such that $\DP(k-1,i',j') \neq \infty$ and determine the minimum $t$ that is reachable. The running time to determine $\DP(k,i,j)$ is thus $\mathcal{O}(n^3)$ by the loops for $i'$, $j'$ and the Fr\'echet distance check. Since there are $\mathcal{O}(kn^2)$ $\DP$-cells to fill, the algorithm runs in total time $\mathcal{O}(kn^5)$ and uses space $\mathcal{O}(kn^2)$. 

\subsection{An $\mathcal{O}(n^3)$ algorithm for Global-Fr\'echet simplification}

Now we improve the running time by a more careful understanding of the monotone paths through $\FS(P,\overline{v_{i'}v_i})$ to $(1,\DP(k,i,j))$ for fixed $i,j$ and $i'$. Let $\mathit{fbox}_j$ denote the intersection of the free-space $\FS(P,\overline{v_{i'}v_i})$ with the square with corner vertices $(0,j)$ and $(1,j+1)$. The following fact will be useful later.

\begin{fact}
 \label{convexity}
 $\mathit{fbox}_j$ is convex for all $j \in [n-1]$.
\end{fact}

\begin{proof}
 Alt and Godau ~\cite{AltG95} showed that $\mathit{fbox}_j$ is an affine transformation of the unit ball, and this is convex for any $L_p$ norm.
\end{proof}

Furthermore let $\mathit{ver}_j$ denote the free space along vertical line segment with endpoints $(0,j)$ and $(0,j+1)$ and let $\mathit{hor}_j$ denote the free space along the horizontal line segment $(0,j)$ to $(1,j)$ in the free space $\FS(P,\overline{v_{i'}v_i})$. We consider the point $(0,j)$ to belong to $\mathit{ver}_j$, but not $\mathit{hor}_j$, to avoid certain corner cases.
We split the monotone paths from $(0,\DP(k-1,i',j'))$ for $i' <i$ and $j' \leq j$ to $(1,\DP(k,i,j))$ in $\FS(P,\overline{v_{i'}v_i})$ into two categories: ones that intersect $\mathit{ver}_j$ and the ones that intersect $\mathit{hor}_j$. We first look at the monotone paths that intersect $\mathit{ver}_j$. Observe that if the monotone path intersects $\mathit{ver}_j$ then $j' =j$.  Let $\overline{\DP}_1(k,i,j) = \min \limits_{i' < i}\DP(k-1,i,j)$. We now define,

$$\DP_1(k,i,j) =    \left\{
\begin{array}{ll}
      \max(\overline{\DP}_1(k,i,j),t_{i,j}) & \text{if }  \overline{\DP}_1(k,i,j) \leq s_{i,j}\\
      \infty &\text{otherwise}
\end{array} 
\right. $$ 

We show a characterization of $\DP_1$ similar to the characterization of $\DP$ in Lemma \ref{recurrence} and thus establishing that $\DP_1$ correctly handles all paths intersecting $\mathit{ver}_j$.

 \begin{observation}
  \label{recurrence-1}
  $\DP_1(k,i,j)$ is the minimal $t \in [t_{i,j},s_{i,j}]$ such that $\DP(k-1,i',j) \neq \infty$ and $\delta_F(P[\DP(k-1,i',j) \ldots t],\overline{v_{i'}v_{i}}) \leq \delta$ for some $i'<i$. If no such $t$ exists then $\DP_1(k,i,j) = \infty$. 
 \end{observation}
 
 \begin{proof}
  \begin{figure}
  \includegraphics[scale = 0.5]{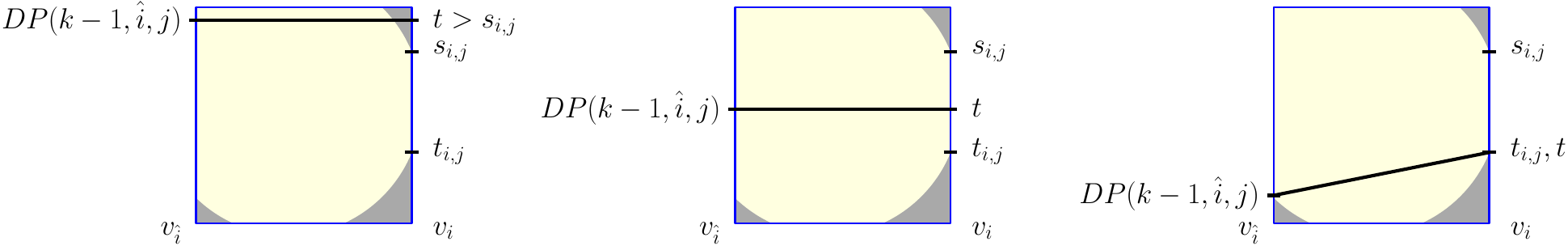}
 \caption{Illustration of the proof of Observation \ref{recurrence-1}. For some $\hat{i} < i$, $t \in [t_{i,j},s_{i,j}]$ is minimal such that $(1,t)$ is reachable from $(0,\DP(k-1,\hat{i},j))$ by a monotone path in $\mathit{fbox}_j$. If $\DP(k-1,\hat{i},j) > s_{i,j}$ (left) then no such $t$ exists. If $t_{i,j} \leq \DP(k-1,\hat{i},j) \leq s_{i,j}$ (middle) then $t = \DP(k-1,\hat{i},j)$. If $\DP(k-1,\hat{i},j) < t_{i,j}$ (right) then $t = t_{i,j}$.}
 \label{reachability-1}
\end{figure}
  
  Fix $\hat{i} < i$. First note that if there is a monotone path connecting $(0,\DP(k-1,\hat{i},j))$ to $(1,t)$ then $t \geq \DP(k-1,\hat{i},j)$. Now consider $\mathit{fbox}_j$ in the free-space $\FS(P,\overline{v_{\hat{i}}v_i})$. As illustrated in Figure \ref{reachability-1} there are three cases,
  \begin{itemize}
    \item If $\DP(k-1,\hat{i},j) > s_{i,j}$ then there is no monotone path from $(0,\DP(k-1,\hat{i},j))$  to $(1,t)$ for all $t \in [t_{i,j},s_{i,j}]$.
        
    \item If $ t_{i,j} \leq \DP(k-1,\hat{i},j) \leq s_{i,j}$. As mentioned at the beginning of the proof, $t \geq \DP(k-1,\hat{i},j)$. Since $\mathit{fbox}_j$ is convex the line segment connecting $(0,\DP(k-1,\hat{i},j)$ and $(1,\DP(k-1,\hat{i},j))$ lies inside $\mathit{fbox}_j$ and hence lies inside $\FS(P,\overline{v_{\hat{i}}v_i})$. Thus the smallest $t \in [t_{i,j},s_{i,j}]$ such that there is a monotone path from $(0,\DP(k-1,\hat{i},j))$  to $(1,t)$ in $\FS(P,\overline{v_{\hat{i}}v_i})$ is $\DP(k-1,\hat{i},j)$. 
     
    \item If $\DP(k-1,\hat{i},j) \leq t_{i,j}$. Again since $\mathit{fbox}_j$ is convex the line segment connecting $(0,\DP(k-1,\hat{i},j))$ and $(1,t_{i,j})$ lies inside $\mathit{fbox}_j$ and thus lies inside $\FS(P,\overline{v_{\hat{i}}v_i})$. Thus the smallest $t \in [t_{i,j},s_{i,j}]$ such that there is a monotone path from $(0,\DP(k-1,\hat{i},j))$  to $(1,t)$ in $\FS(P,\overline{v_{\hat{i}}v_i})$ is $t_{i,j}$.
  \end{itemize}

Therefore, for any $\hat{i} < i$ if $\DP(k-1,\hat{i},j) > s_{i,j}$ then there exists no $t \in [t_{i,j},s_{i,j}]$ such that $\delta_F(P[\DP(k-1,\hat{i},j) \ldots t],\overline{v_{\hat{i}}v_i}) \leq \delta$. Similarly if $\DP(k-1,\hat{i},j) \leq s_{i,j}$  then the minimal $t \in [t_{i,j},s_{i,j}]$ such that $\delta_F(P[\DP(k-1,\hat{i},j) \ldots t],\overline{v_{\hat{i}}v_i}) \leq \delta$ is $\max(\DP(k-1,\hat{i},j),t_{i,j})$.\\
Now let $t \in [t_{i,j},s_{i,j}]$ be minimal such that $\DP(k-1,i',j) \neq \infty$ and $\delta_F(P[\DP(k-1,i',j) \ldots t],$ $\overline{v_{i'}v_{i}}) \leq \delta$ for some $i'<i$. It follows that if $\min \limits_{i' < i} \DP(k-1,i',j) \leq s_{i,j}$, then $t = \max( \min \limits_{i' < i} \DP(k-1,i',j),t_{i,j})$ and if $\min \limits_{i' < i} \DP(k-1,i',j) > s_{i,j}$, then no such $t$ exists. Since $\min \limits_{i' < i} \DP(k-1,i',j)  = \overline{\DP}_1(k,i,j)$ and $\DP_1(k,i,j) = \max(\overline{\DP}_1(k,i,j),t_{i,j})$ when $\overline{DP}_1(k,i,j) \leq s_{i,j}$ (by definition), we have that when $\overline{\DP}_1(k,i,j) \leq s_{i,j}$,  $\DP_1(k,i,j) = \max(\overline{\DP_1}(k,i,j),t_{i,j}) = \max(\min \limits_{i' < i} \DP(k-1,i',j),t_{i,j}) = t$. Similarly when $\overline{\DP}_1(k,i,j) > s_{i,j}$, then $DP(k,i,j)= \infty$ and $t$ does not exist.
 \end{proof}
 
  We now look at the monotone paths that intersect $\mathit{hor}_j$. Observe that if the monotone path intersects $\mathit{hor}_j$ then $j' < j$. Along this line, we define $\overline{\DP}_2(k,i,j) = 1$ if there exists some $i' <i $ and $j' < j$, such that $ \DP(k-1,i',j') \neq \infty$  and there exists a monotone path from  $(0,\DP(k-1,i',j'))$ to $(1,t_{i,j})$ in  the free-space $\FS(P,\overline{v_{i'}v_{i}})$ and otherwise we set $\overline{\DP}_2(k,i,j)= 0$. Hereafter we define,

$$\DP_2(k,i,j) =    \left\{
\begin{array}{ll}
      t_{i,j} & \text{if }  \overline{\DP}_2(k,i,j) =1\\
      \infty &\text{otherwise}
\end{array} 
\right. $$ 

We show a characterization of $\DP_2$ similar our characterization of $\DP$ in Lemma \ref{recurrence}, and thus establishing that $\DP_2$ correctly handles all paths intersecting $\mathit{hor}_j$. 
\begin{observation}
 \label{recurrence-2}
 $\DP_2(k,i,j)$ is the minimal $t \in [t_{i,j},s_{i,j}]$ such that $\DP(k-1,i',j') \neq \infty$ and $\delta_F(P[\DP(k-1,i',j') \ldots t], \overline{v_{i'}v_i}) \leq \delta$ for some $i'<i$ and $j' < j$. If no such $t$ exists then $\DP_2(k,i,j) = \infty$.
\end{observation}

\begin{proof}
 
  \begin{figure}
 \begin{subfigure}{0.5\textwidth}
   \includegraphics[scale=0.5]{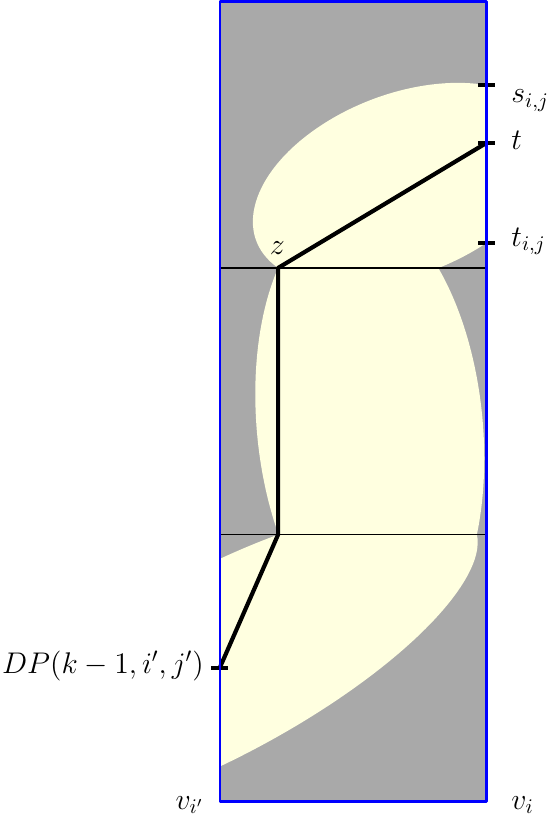}
 \end{subfigure}
 \begin{subfigure}{0.5\textwidth}
   \includegraphics[scale=0.5]{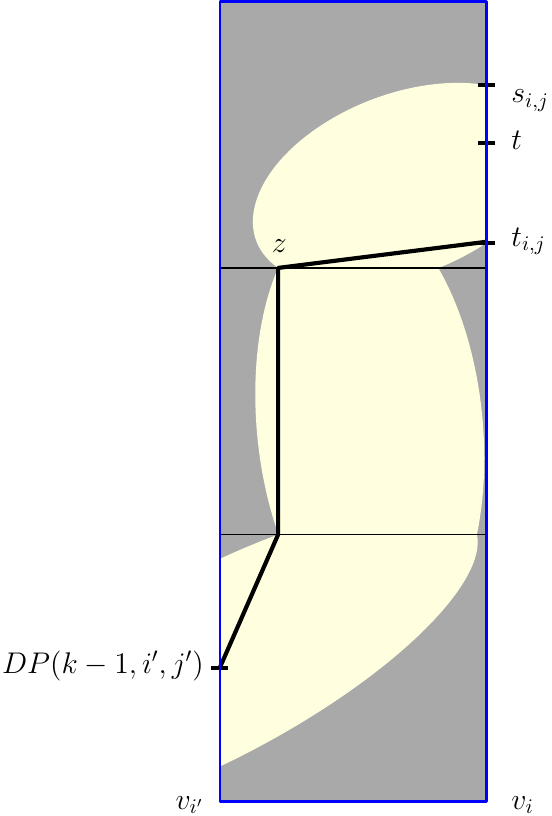}
 \end{subfigure}
 \caption{Illustration of the proof of Observation \ref{recurrence-2}. For $t_{i,j} \leq t \leq s_{i,j}$, there is a monotone path from $(0,\DP(k-1,i',j'))$ to $(1,t)$ in the free-space $\FS(P,\overline{v_{i'}v_i})$ (left) for some $i' < i$ and $j' < j$ that intersect $\hor_j$ at $z$. Then there is also a monotone path from $(0,\DP(k-1,i',j'))$ to $(1,t_{i,j})$ (right) in the free-space $\FS(P,\overline{v_{i'}v_i})$  following the same monotone path from $(0,\DP(k-1,i',j')$ to $z$ and then from $z$ to $(1,t_{i,j})$.}
 \label{reachability-2}
\end{figure}

 Let $t \in [t_{i,j},s_{i,j}]$ be minimal such that $\DP(k-1,i',j') \neq \infty$ and $\delta_F(P[\DP(k-1,i',j') \ldots t],$ $ \overline{v_{i'}v_i}) \leq \delta$ for some $i'<i$ and $j' < j$. If such a $t$ exists then $\overline{DP}_2(k,i,j)=1$. Observe that for any $i' < i$ and $j' < j$, if there is a monotone path from $(0,\DP(k-1,i',j'))$ to $(1,t)$ in $\FS(P,\overline{v_{i'}v_i})$, then the path intersects $\mathit{hor}_j$ (at say $z$). Since $\mathit{fbox}_j$ is convex, the line segment connecting $z$ and $(1,t_{i,j})$ lies inside $\mathit{fbox}_j$ and hence inside $\FS(P,\overline{v_{i'}v_i})$. Thus there is a monotone path from $(0,\DP(k-1,i',j'))$ to $(1,t_{i,j})$ in $\FS(P,\overline{v_{i'}v_i})$ following the monotone path from $(0,\DP(k-1,i',j'))$ to $z$ and then from $z$ to $(1,t_{i,j})$ (see Figure \ref{reachability-2}). Since $t \geq t_{i,j}$ and is minimal, we have $t = t_{i,j} = \DP_2(k,i,j)$. Similarly if such a $t$ does not exist then $\overline{\DP}_2(k,i,j) =0$ and $\DP_2(k,i,j)=\infty$.   
\end{proof}
 
\begin{lemma}
 \label{First-formulation}
 $\DP(k,i,j) = \min (\DP_1(k,i,j), \DP_2(k,i,j))$.
\end{lemma} 

\begin{proof}
 Follows directly from Observations \ref{recurrence}, \ref{recurrence-1}, and \ref{recurrence-2}.
\end{proof}
In particular this yields a dynamic programming formulation for $\DP(k,i,j)$, since both $\DP_1(k,i,j)$ and $\DP_2(k,i,j)$ depends on values of $\DP(k',i',j')$ with $k' < k$, $i' < i$ and $j' \leq j$.

We define $\kappa(i,j)$ as the minimal $k$ such that $\DP(k,i,j) \neq\infty$. Similarly we define $\kappa_1(i,j)$ and  $\kappa_2(i,j)$ as the minimal $k$ such that ${\DP}_1(k,i,j) \neq \infty$ and $\DP_2(k,i,j) \neq \infty$ respectively. Note that $\kappa(i,j) = \min(\kappa_1(i,j),\kappa_2(i,j))$ (by Lemma \ref{First-formulation}). Also note that both $\kappa_1(i,j)$ and $\kappa_2(i,j)$ depends only on the values of $\DP(k',i',j')$ with $k' < k$, $i' < i$ and $j' \leq j$.     

With these preparations can now present our dynamic programming algorithm, except for one subroutine \emph{$\kappa_2$-subroutine$(i)$} that we describe in Section \ref{kappa2subroutine}. In particular, for any $i$, $\kappa_2$-subroutine$(i)$ determines $\kappa_2(i,j)$ for all $j \in [n]$ in time $T(n)$ only using the values of $\kappa(i',j)$ for all $i' < i$ and all $0 \leq j \leq n-1$. Now we show how to update $\DP_1(k,i,j)$. Observe that for any $i$, $j$ and $k$ we can update $\overline{\DP}_1(k,i,j)$ from $\overline{\DP}_1(k,i-1,j)$ and $\DP(k-1,i-1,j)$ as $\overline{\DP}_1(k,i,j) = \min(\overline{\DP}_1(k,i-1,j),\DP(k-1,i-1,j))$. Thereafter we can update $\DP_1(k,i,j)$ by using the formulation in Lemma \ref{recurrence-1} and update $\kappa_1(i,j)$ to the minimal $k$ such that $\DP_1(k,i,j) \neq \infty$. This shows that we determine $\overline{\DP}_1(k,i,j)$ and $\kappa_1(i,j)$\ in $\mathcal{O}(1)$ and $\mathcal{O}(n)$ time respectively. Now we show how to update $\DP_2(k,i,j)$. Notice that $\DP_2(k,i,j) =t_{i,j}$ if and only if $k \geq \kappa_2(i,j)$ and $\DP_2(k,i,j) = \infty$ otherwise. Also, we can set $\kappa(i,j)$ as $\min(\kappa_1(i,j),\kappa_2(i,j))$. Hence, we can determine $\DP_2(k,i,j)$ and  $\kappa(i,j)$ in $\mathcal{O}(1)$ time. Henceforth we can also update $\DP(k,i,j)$ by the formulation in Lemma \ref{First-formulation} in $\mathcal{O}(1)$ time.

\begin{algorithm}[H]
 \begin{algorithmic}[1]
  \State \textbf{Determine} $t_{i,j}$ and $s_{i,j}$ for all $0 \leq i \leq n$ and $0 \leq j \leq n-1$
  \State \textbf{Determine} the largest $j_0$ such that $\norm{v_0-v_j}{p} \leq \delta$ for all $j\leq j_0$
  \State \textbf{Set} $\overline{\DP}_1(k,0,j),\DP(k,0,j)$ to $0$ for all $j \leq j_0$ and to $\infty$ otherwise (for all  $k \in [n+1]$)
  \State \textbf{Set} $\kappa(0,j)$ to $1$ for all $j \leq j_0$ and to $\infty$ otherwise
  \State \textbf{Set} $\DP(0,i,j)$ to $\infty$ for all $i,j \in [n]$
  \For{$i=1$ to $n$}
      \State \emph{Determine $\kappa_2(i,j)$ for all $0 \leq j \leq n-1$ using $\kappa_2$-subroutine$(i)$}
      \For{$j=0$ to $n-1$}
        \For{$k= 1$ to $n+1$}
           \State \textbf{Set} $\overline{\DP}_1(k,i,j)$ to   $\min(\overline{\DP}_1(k,i-1,j),\DP(k-1,i-1,j))$
           \State \textbf{Set} ${\DP}_1(k,i,j)$ to  $\max(\overline{\DP}_1(k,i,j),t_{i,j})$ if $\overline{\DP}_1(k,i,j) \leq s_{i,j}$ and to $\infty$ otherwise 
        \EndFor
        \State \textbf{Set} $\kappa_1(i,j)$ to the smallest $k$ such that $\DP_1(k,i,j) \neq \infty$
        \State \textbf{Set} $\kappa(i,j) = \min(\kappa_1(i,j),\kappa_2(i,j))$
        
        \For{$k= 1$ to $n+1$}
            \State \textbf{Set} $\DP_2(k,i,j)$ to $t_{i,j}$ if $k \geq \kappa_2(i,j)$ and to $\infty$ otherwise
            \State \textbf{Set} ${\DP}(k,i,j)$ to $\min({\DP}_1(k,i,j),\DP_2(k,i,j))$
         \EndFor
     \EndFor
   \EndFor
   \State \textbf{Return} $\kappa(n,n-1)$
 \end{algorithmic}
 \caption {Solving curve simplification under Global-Fr\'echet distance}
 \label{algorithm-1}
\end{algorithm}

Algorithm \ref{algorithm-1} takes $\mathcal{O}(n \cdot T(n))$ time for  determining $\kappa_2(i,j)$ for all $i,j \in [n]$. The time taken to update $\kappa_1(i,j)$ and $\kappa(i,j)$ is $\mathcal{O}(n)$ and $\mathcal{O}(1)$ respectively. All the $\DP$ cells are updated in $\mathcal{O}(1)$ time. Since there are $\mathcal{O}(n^2)$ $\kappa$ cells and $\mathcal{O}(n^3)$ $\DP$ cells, the total running time of our algorithm is $\mathcal{O}(n^3 + n\cdot T(n))$.


\subsection{Implementing $\kappa_2$-subroutine$(i)$}
\label{kappa2subroutine}
In this subsection we show how to implement step 7 of Algorithm \ref{algorithm-1} in time $T(n) = \mathcal{O}(n^2)$. Then in total we have $\mathcal{O}(n^3)$ for solving Global-Fr\'echet simplification. 
\subsubsection{Cell Reachability}
We introduce an auxiliary problem that we call \emph{Cell Reachabilty}. We shall see later that an $\mathcal{O}(n)$ time solution to this problem ensures that the $\kappa_2$-subroutine$(i)$ can be implemented in  time $T(n) = \mathcal{O}(n^2)$.

\begin{definition}
In an instance of the {Cell Reachabilty} problem, we are given 
\begin{itemize}
    \item A set of $n$ \emph{cells}. Each cell $j$ with $1\leq j \leq n$ is a unit square with corner points $(0,j)$ and $(1,j+1)$. We say that cells $j$ and ${j+1}$ are \emph{consecutive}. 
    \item An integral \emph{entry-cost} $\lambda_j > 0$ for every cell $j$.
    \item A set of $n-1$ \emph{passages} between consecutive cells. The passage $p_j$ is the horizontal line segment with endpoints $(j,a_j)$ and $(j,b_j)$ where $b_j > a_j$.
\end{itemize}
 
  We say that a cell  $j$ is \emph{reachable} from a cell ${j'}$ with $j' < j$ if and only if there exists $ x_{j'+1} \leq x_{j'+2}  \ldots \leq x_{j}$ such that $x_k \in [a_k,b_k]$ for every $j' < k \leq j$. Intuitively cell $j$ is reachable from cell ${j'}$ if and only if there is a monotone path through the passages from cell ${j'}$ to cell $j$. We define the \emph{exit-cost} $\mu_j$ of a cell $j$ as the minimal $\lambda_{j'}$ such that $j$ is reachable from cell ${j'}$, $j' < j$. {The goal of the problem is to determine the sequence $\langle \mu_1,\mu_2, \ldots ,\mu_n \rangle$.} See Figure \ref{Cellreachability1} for an illustration. 
\end{definition}

\begin{figure}
  \centering
 \begin{subfigure}{0.3\textwidth}

\begin{tikzpicture}

\draw[blue, very thick, fill = green!5] (0,0) rectangle (2,8);

\draw (0,2) -- (2,2);
\draw (0,4) -- (2,4);
\draw (0,6) -- (2,6);

\node[above =0pt of {(0.4,7)}]{$ \scriptstyle \mathit{Cell} \text{ } 4$};
\node[above =0pt of {(0.4,5)}]{$ \scriptstyle \mathit{Cell} \text{ } 3$};
\node[above =0pt of {(0.4,3)}]{$ \scriptstyle \mathit{Cell} \text{ } 2$};
\node[above =0pt of {(0.4,0.5)}]{$ \scriptstyle \mathit{Cell} \text{ } 1$};


\draw[red, ultra thick, ->] (0-0.2,1)--(0,1);
 \node[left = 0pt of {(0-0.2,1)}]{$\scriptstyle 1 = \lambda_{1}$};

\draw[red, ultra thick, ->] (0-0.2,3)--(0,3);
 \node[left = 0pt of {(0-0.2,3)}]{$\scriptstyle 4 = \lambda_{2}$};
 
\draw[red, ultra thick, ->] (0-0.2,5)--(0,5);
 \node[left = 0pt of {(0-0.2,5)}]{$\scriptstyle 8 = \lambda_{3}$};


\draw[dashed] (0,1) -- (1.2,2) -- (1.2,4) -- (1.2,6);
\draw[black, ultra thick] (1.2-0.1,6-0.1) -- (1.2+0.1,6+0.1);
\draw[black, ultra thick] (1.2-0.1,6+0.1) -- (1.2+0.1,6-0.1);

\draw[red, ultra thick] (1.2,2)--(1.7,2);
\draw[red, ultra thick] (0.6,4)--(1.3,4);
\draw[red, ultra thick] (0.2,6)--(0.8,6);

\end{tikzpicture}
 \end{subfigure}
 \begin{subfigure}{0.3\textwidth}

\begin{tikzpicture}

\draw[blue, very thick, fill = green!5] (0,0) rectangle (2,8);

\draw (0,2) -- (2,2);
\draw (0,4) -- (2,4);
\draw (0,6) -- (2,6);

\node[above =0pt of {(0.4,7)}]{$ \scriptstyle \mathit{Cell} \text{ } 4$};
\node[above =0pt of {(1.4,5)}]{$ \scriptstyle \mathit{Cell} \text{ } 3$};
\node[above =0pt of {(1.4,3)}]{$ \scriptstyle \mathit{Cell} \text{ } 2$};
\node[above =0pt of {(0.4,0.5)}]{$ \scriptstyle \mathit{Cell} \text{ } 1$};


\draw[red, ultra thick, ->] (0-0.2,1)--(0,1);
 \node[left = 0pt of {(0-0.2,1)}]{$\scriptstyle 1 = \lambda_{1}$};

\draw[red, ultra thick, ->] (0-0.2,3)--(0,3);
 \node[left = 0pt of {(0-0.2,3)}]{$\scriptstyle 4 = \lambda_{2}$};
 
\draw[red, ultra thick, ->] (0-0.2,5)--(0,5);
 \node[left = 0pt of {(0-0.2,5)}]{$\scriptstyle 8 = \lambda_{3}$};


\draw[dashed] (0,3) -- (0.7,4) -- (0.7,6);

\draw[red, ultra thick] (1.2,2)--(1.7,2);
\draw[red, ultra thick] (0.6,4)--(1.3,4);
\draw[red, ultra thick] (0.2,6)--(0.8,6);

\end{tikzpicture}
 \end{subfigure}
 \begin{subfigure}{0.3\textwidth}

\begin{tikzpicture}

\draw[blue, very thick, fill = green!5] (0,0) rectangle (2,8);

\draw (0,2) -- (2,2);
\draw (0,4) -- (2,4);
\draw (0,6) -- (2,6);

\node[above =0pt of {(0.4,7)}]{$ \scriptstyle \mathit{Cell} \text{ } 4$};
\node[above =0pt of {(0.6,5)}]{$ \scriptstyle \mathit{Cell} \text{ } 3$};
\node[above =0pt of {(0.4,3)}]{$ \scriptstyle \mathit{Cell} \text{ } 2$};
\node[above =0pt of {(0.4,0.5)}]{$ \scriptstyle \mathit{Cell} \text{ } 1$};


\draw[red, ultra thick, ->] (0-0.2,1)--(0,1);
 \node[left = 0pt of {(0-0.2,1)}]{$\scriptstyle 1 = \lambda_{1}$};

\draw[red, ultra thick, ->] (0-0.2,3)--(0,3);
 \node[left = 0pt of {(0-0.2,3)}]{$\scriptstyle 4 = \lambda_{2}$};
 
\draw[red, ultra thick, ->] (0-0.2,5)--(0,5);
 \node[left = 0pt of {(0-0.2,5)}]{$\scriptstyle 8 = \lambda_{3}$};


\draw[dashed] (0,5) -- (0.2,6);

\draw[red, ultra thick] (1.2,2)--(1.7,2);
\draw[red, ultra thick] (0.6,4)--(1.3,4);
\draw[red, ultra thick] (0.2,6)--(0.8,6);

\end{tikzpicture}
 \end{subfigure} 
 \caption{Illustrating an instance of Cell Reachability. The red horizontal line segments between the cells indicate the passages. Note that cell 4 is only reachable from cells 2 and 3. therefore $\mu_4 = \min(\lambda_2,\lambda_3) = \min(4,8) = 4$.}
 \label{Cellreachability1}
\end{figure}
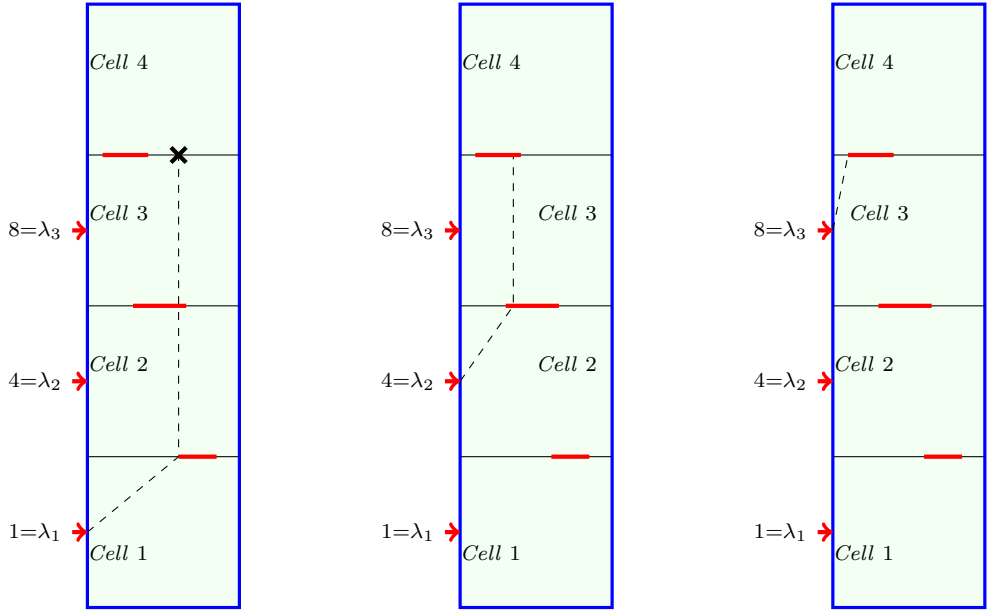

  We make a more refined notion of reachability. For any cells $j$ and $j' < j$ we define the \emph{first reachable point} $\frp(j,j')$ on cell $j$ from cell $j'$ as the minimal $t$ such that there exists $x_{j'+1} \leq x_{j'+2} \leq  \ldots \leq x_j$ such that each $x_k \in [a_k,b_k]$ for every $j' < k \leq j$ and $x_j = t$ and we set $\frp(j,j') = \infty$ if there exists no such $t$. Let $t_j(k)$ be the first reachable point on cell $j$ from any cell $j'$ with entry-cost at most $k$ i.e. $t_j(k) = \min \limits_{j' < j, \lambda_{j'} \leq k} \frp(j,j')$. In particular we have $t_j(0) = \infty$, since $\lambda_{j'} > 0$ for all $j' <j$. We now make some simple observations about $t_j(k)$.

\begin{observation}
  \label{sufficiency}
    $\mu_j$ is the minimal $k$ such that $t_j(k) \neq \infty$.
  \end{observation}
  \begin{proof}
    We have $t_j(k) \neq \infty$ if and only if cell $j$ is reachable from some cell $j' < j$ with entry-cost $\lambda_{j'} \leq k$. Therefore the minimal such $\lambda_{j'}$ is the minimal $k$ at which $t_j(k) \neq \infty$. 
  \end{proof}
  
  \begin{observation}
    \label{monotonicity}
    We have $t_j(k+1) \leq t_j(k)$ for any $j \in [n]$ and $k \geq 0$. 
  \end{observation}
  \begin{proof}
    The minimum in the definition of $t_j(k+1)$ is taken over a superset compared to $t_j(k)$.
  \end{proof}

  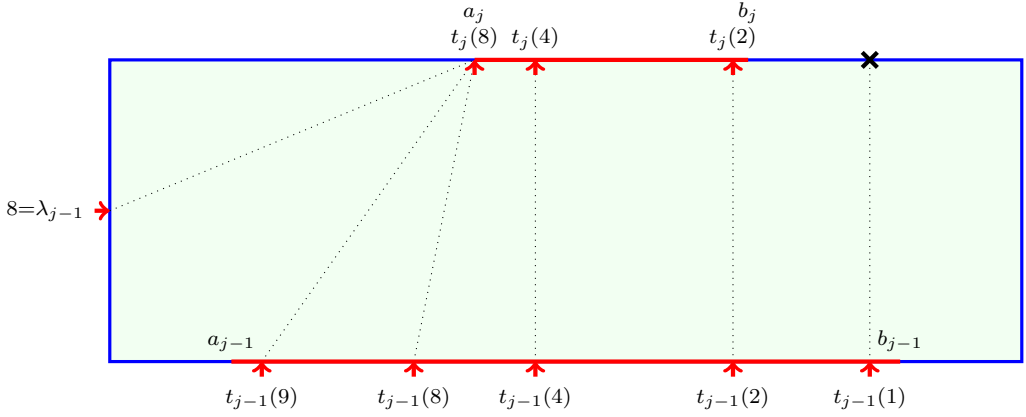
\begin{figure}
    \centering

\begin{tikzpicture}

\draw[blue, very thick, fill = green!5] (0,0) rectangle (12,4);

\draw[red, ultra thick, ->] (0-0.2,2)--(0,2);
 \node[left = 0pt of {(0-0.2,2)}]{$\scriptstyle 8 = \lambda_{j-1}$};


\draw[dotted] (0,2) -- (2*2.4,4);
\draw[dotted] (2,0) -- (2*2.4,4);
\draw[dotted] (4,0) -- (2*2.4,4);

\draw[dotted] (2*2.8,0) -- (2*2.8,4);
\draw[dotted] (2*4.1,0) -- (2*4.1,4);
 
\draw[dotted] (10,0) -- (10,4);

\draw[red, ultra thick] (1.6,0)--(10.4,0);
 \node[above = 0pt of {(2*0.8,0)}]{$\scriptstyle a_{j-1}$};
  \node[above = 0pt of {(2*5.2,0)}]{$\scriptstyle b_{j-1}$};

\draw[red, ultra thick] (2*2.4,4)--(2*4.2,4);
 \node[above = 0pt of {(2*2.4,4+0.35)}]{$\scriptstyle a_{j}$};
  \node[above = 0pt of {(2*4.2,4+0.35)}]{$\scriptstyle b_{j}$};


\draw[red, ultra thick, ->] (2*1,0-0.2) -- (2*1,0);
 \node[below = 0pt of {(2*1,0-0.2)}]{$\scriptstyle t_{j-1}(9)$};

\draw[red, ultra thick, ->] (2*2,0-0.2) -- (2*2,0);
 \node[below = 0pt of {(2*2,0-0.2)}]{$\scriptstyle t_{j-1}(8)$};

\draw[red, ultra thick, ->] (2*2.8,0-0.2) -- (2*2.8,0);
 \node[below = 0pt of {(2*2.8,0-0.2)}]{$\scriptstyle t_{j-1}(4)$};

\draw[red, ultra thick, ->] (2*4.1,0-0.2) -- (2*4.1,0);
  \node[below = 0pt of {(2*4.1,0-0.2)}]{$\scriptstyle t_{j-1}(2)$};

\draw[red, ultra thick, ->] (2*5,0-0.2) -- (2*5,0);
  \node[below = 0pt of {(2*5,0-0.2)}]{$\scriptstyle t_{j-1}(1)$};

\draw[black, ultra thick] (2*5-0.1,4-0.1) -- (2*5+0.1,4+0.1);
\draw[black, ultra thick] (2*5-0.1,4+0.1) -- (2*5+0.1,4-0.1);


\draw[red, ultra thick, ->] (2*2.4,4-0.2) -- (2*2.4,4);
 \node[above = 0pt of {(2*2.4,4)}]{$\scriptstyle t_{j}(8)$};
  
\draw[red, ultra thick, ->] (2*2.8,4-0.2) -- (2*2.8,4);
  \node[above = 0pt of {(2*2.8,4)}]{$\scriptstyle t_{j}(4)$}; 
 
\draw[red, ultra thick, ->] (2*4.1,4-0.2) -- (2*4.1,4);
   \node[above = 0pt of {(2*4.1,4)}]{$\scriptstyle t_{j}(2)$};

\end{tikzpicture}
    \caption{Illustration of the proof of Lemma \ref{recurrencefort}. Determining the function $t_j(\cdot)$ from $a_j$,$b_j$,$t_{j-1}(\cdot)$, and $\lambda_{j-1}$. For all $k \geq \lambda_{j-1} = 8$ we have $t_j(k) = a_j$. For $k =2$ and $k=4$ we have $k< \lambda_{j-1}$ and $t_{j-1}(k) \leq b_j$, implying $t_j(k) = t_j(k-1)$. Lastly for $k=1$ we have $t_{j-1}(k) > b_j$, implying $t_j(k) = \infty$.}
    \label{Cellreachability2}
  \end{figure}

 \begin{lemma}
 \label{recurrencefort}
  For any $j \in [n]$ and $k \geq 0$ we have 
 \begin{align*}
  t_j(k) &= \left\{ \begin{array}{ll}
          				a_j & \text{if }  k \geq \lambda_{j-1}\\
          				a_j & \text{if }   k < \lambda_{j-1} \text{ and } t_{j-1}(k) \leq a_j\\
          				t_{j-1}(k) & \text{if } k < \lambda_{j-1} \text{ and } t_{j-1}(k) \in (a_j,b_j]\\
          				\infty & \text{if } k < \lambda_{j-1} \text{ and } t_{j-1}(k) > b_j
          			\end{array}  
          			\right. 
  \end{align*}       
\end{lemma}

\begin{proof}
  See Figure \ref{Cellreachability2} for an illustration. Note that $\frp(j,j-1) = a_j$.  Therefore if $\lambda_{j-1} \leq k$, then $t_j(k) = \min \limits_{j' < j, \lambda_{j'} \leq k} \frp(j,j') \leq \frp(j,j-1) = a_j $. Since $t_j(k) \geq a_j$, we conclude that $t_j(k)=a_j$. Now we discuss the cases when $k < \lambda_{j-1}$. Let $t_j(k) = \frp(j,j')$. Since $\lambda_{j-1} > k$ we have $j' < j-1$. Therefore there exist $x_{j'+1} \leq x_{j'+2} \leq \ldots \leq  x_{j-1} \leq x_j$ such that $x_k \in [a_k,b_k]$ for every $j' < k \leq j$ with $x_j = t_j(k)$. Note that $t_{j-1}(k) \leq x_{j-1} \leq x_j = t_j(k)$. Thus $t_j(k) \geq \max(t_{j-1}(k),a_j)$. In particular, if $t_{j-1}(k) > b_j$, then $t_j(k) = \infty$. Now we look into the case when $t_{j-1}(k) \leq b_j$. Observe that if $t_{j-1}(k) \leq b_j$ then there exists $\hat{j} < j-1$ and there exists $ x_{\hat{j}+1} \leq x_{\hat{j}+2} \leq  \ldots \leq x_{j-1} = t_{j-1}(k)$ such that $x_k \in [a_k,b_k]$ for every $\hat{j} < k \leq j-1$. Setting $x_j = \max(a_j,t_{j-1}(k))$ and there exists $ x_{\hat{j}+1} \leq x_{\hat{j}+2} \leq  \ldots \leq x_{j-1} \leq x_j$ such that $x_k \in [a_k,b_k]$ for every $\hat{j} \leq k \leq j$ and hence $t_j(k) \leq \max(a_j,t_{j-1}(k))$. Combining the two inequalities we get that $t_j(k) = \max(a_j,t_{j-1}(k))$ when $t_{j-1}(k) \leq b_j$.
\end{proof}

Lemma \ref{recurrencefort} yields a recursive definition for $t_j( \cdot )$. To ensure that we can solve an instance of \emph{cell reachability} in $\mathcal{O}(n)$ time, if suffices to determine $t_{j}(\cdot)$ from $t_{j-1}(\cdot)$ and $\mu_j$ from $t_j(\cdot)$ in $\mathcal{O}(1)$ amortized time. 
To this end, let $S_j = \set{k \geq 0}{t_j(k) < t_j(k-1)}$ and let $L_j$ be the \emph{doubly linked list} storing the pairs $(k,t_j(k))$ for every $k \in S_j$, sorted in descending order of $k$ (or equivalently in increasing order of $t_j(k)$). To develop some intuition note that for any $k$ and $j$ if we have $t_j(k) = t_j(k-1)$, then this means that every cell $j' \geq j$ that is reachable from a cell $\hat{j} \leq j$ with entry-cost at most $k$ is also reachable from some cell $\tilde{j} \leq j$ with entry-cost at most $k-1$. Since we are only interested in reachability from a cell of minimum entry-cost, we can ignore reachability from all cells below cell $j$ with entry costs $k$. Therefore it suffices to focus on the set $S_j$ and the corresponding $\mu_j$. In particular we can determine $\mu_j$ from $S_j$ as following,    
 
 \begin{lemma}
   The minimal positive $k$ in $S_j$ is equal to $\mu_j$.
 \end{lemma}
 \begin{proof}
  Since $t_j(0) = \infty$, the minimal positive $k$ in $S_j$ is the minimal $k$ such that $t_j(k) \neq \infty$. By Observation \ref{sufficiency} this is equal to $\mu_j$. 
 \end{proof}
 
 We now outline a simple algorithm to determine $L_{j}$ from $L_{j-1}$. Again see Figure \ref{Cellreachability2} for illustration. The algorithm first determines $k_{\mathit{left}}$, the minimal $k$ such that $t_j(k) = a_j$, by moving the head of the list $L_{j-1}$ to the right as long as $k \geq \lambda_{j-1}$ or $t_{j-1}(k) \leq a_j$ (correctness follows directly from Lemma \ref{recurrencefort}). Observe that $t_j(k) = t_j(k_{\mathit{left}}) = a_j$ for all $k \geq k_{\mathit{left}}$.  Next it determines $k_{\mathit{right}}$, the minimal $k$ such that $t_j(k) \leq b_j$ by moving the tail of $L_{j-1}$ to the minimal $k$ such that $t_{j-1}(k) \leq b_j$. Note that at this point we have already inserted $(k_{\mathit{left}},a_j)$ so $k_{\mathit{right}}$ is guaranteed to exits.(Again correctness follows from Lemma \ref{recurrencefort}). Observe that $t_j(k) = t_j(0) = \infty$ for all $k < k_{\mathit{right}}$. Thus we have $\mu_j = k_{\mathit{right}}$ . Now we are left with updating $L_j$ for pairs with $k \in (k_{\mathit{left}}, k_{\mathit{right}})$. Note that for $k \in (k_{\mathit{left}}, k_{\mathit{right}})$, we have $t_j(k) = t_{j-1}(k)$ (by Lemma \ref{recurrencefort}) and therefore $t_j(k) = t_j(k-1)$ if and only if $t_{j-1}(k) = t_{j-1}(k-1)$. Thus the sublist of $L_{j}$ corresponding to the values of $k \in (k_{\mathit{left}},k_{\mathit{right}})$ is same as the sublist of $L_{j-1}$ corresponding to the values of $k \in (k_{\mathit{left}},k_{\mathit{right}})$. Finally the algorithm appends a new node to $L_j$ storing $(0,\infty)$ (since $t_j(0)= \infty$).          
 
 \begin{algorithm}[H]
 \begin{algorithmic}[1]
    \State $L \gets L_{j-1}$
    \State $k_{\mathit{left}} \gets \lambda_j$
    \While{$k \geq \lambda_j$ or $t \leq a_j$, where $(k,t)= L.\mathit{front}()$}
      \State $k_\mathit{left} \gets \min(k_{\mathit{left}},k)$
      \State $L.\mathit{popfront}()$
     \EndWhile
     \State $L.\mathit{pushfront}((k_\mathit{left},a_j)$
     \While{$t > b_j$, where $(k,t) = L.\mathit{back}()$}
       \State $L.\mathit{popback}()$
     \EndWhile
     \State \textbf{Set} $\mu_j = k$, where $(k,t) = L.\mathit{back}()$.
     \State $L.\mathit{pushback}((0,\infty))$
     \State $L_j \gets L$
 \end{algorithmic}
 \caption {Determining $L_{j}$ from $L_{j-1}$}
\end{algorithm}
 
The number of operations performed to determine $L_j$ from $L_{j-1}$ and determining $\mu_j$ from $L_j$ is $\mathcal{O}(1 + d)$ where $d$ is the number of pairs deleted from $L_{j-1}$. Since every deleted pair was previously inserted, we can pay for the deletions by paying an extra token per insertion. Note that there are two insertions per update. Hence the total time taken to determine $L_j$ and $\mu_j$ for all $j \in [n]$ is $\mathcal{O}(n)$.

\begin{theorem}
 {Cell Reachability} can be solved in $\mathcal{O}(n)$ time.
\end{theorem} 
 
\subsubsection{Implementing $\kappa_2$-subroutine$(i)$ using Cell Reachability} 
Recall the definition of $\kappa_2(\cdot ,\cdot )$ and what our goal is now. For a fixed $i' < i$, let  $\kappa(i,j,i')$ be the minimal $k$ such that for some $j' <j$, we have $\DP(k-1,i',j') \neq \infty$ and $\delta_{F}(P[\DP(k-1,i',j') \ldots t_{i,j}],\overline{v_{i'}v_{i}}) \leq \delta$. Note that $\kappa_2(i,j) = \min \limits_{i' < i} \kappa(i,j,i')$. To show that the $\kappa_2$-subroutine$(i)$ can be implemented in $\mathcal{O}(n^2)$, it suffices to show that for a fixed $i' < i$ we can determine $\kappa(i,j,i')$ for all $j \in [n-1]$  in $\mathcal{O}(n)$ time.

\begin{observation}
 \label{monotonicity-reachability}
 Let the line segment with endpoints $(a_j,j)$ and $(b_j,j)$ denote the free-space on $\mathit{hor}_j$ in $\FS(P,\overline{v_{i'}v_i})$ where $i' < i$. Then for any $j' < j$ there is a monotone path from $(0,\DP(\kappa(i',j'),i',j'))$ to $(1,t_{i,j})$ in the free-space $\FS(P,\overline{v_{i'}v_i})$ if and only if there exist $x_{j'+1} \leq x_{j'+2} \leq \ldots x_j$ with each $x_k \in [a_k,b_k]$ for all $j' < k \leq j$.  
\end{observation}

\begin{proof}
 The \say{only if} direction is straightforward. Note that the monotone path from $(0,\DP$ $(\kappa(i',j'),i',j'))$ to $(1,t_{i,j})$ in the free-space $\FS(P,\overline{v_{i'}v_i})$ intersects $\hor_k$ for all $j' < k \leq j$. Let $x_k$ be the intersection of the path with $\hor_k$ for $ j' < k \leq j$. Since the path lies inside the free-space $\FS(P,\overline{v_{i'}v_i})$ we have $x_k \in [a_k,b_k]$ for every $j' < k \leq j$. Since the path is monotone we have $x_{j'+1} \leq x_{j'+2} \leq \ldots \leq x_{j}$.\\
Now we show the \say{if} direction. Assume there exist $x_{j'+1} \leq x_{j'+2} \leq \ldots \leq x_{j}$  and $x_k \in [a_k,b_k]$ for every $ j' < k \leq j$. Since every $\mathit{fbox}_k$ is convex for every $j' < k < j$, the line segment with endpoints as $(x_k,k)$ and $(x_{k+1},k+1)$ lies inside $\mathit{fbox}_k$. By the same convexity argument it follows that the line segment with endpoints $(0,\DP(\kappa(i',j'),i',j')$ and $(x_{j'+1},j'+1)$ lies inside $\mathit{fbox}_{j'}$ and the line segment with endpoints $(x_j,j)$ and $(1,t_{i,j})$ also lies inside $\mathit{fbox}_j$. Therefore we have a monotone path namely $\langle (0,\DP(\kappa(i',j'),i',j'), (x_{j'+1},j'+1), (x_{j'+2},j'+2) \ldots (x_j,j) (1,t_{i,j}) \rangle$ inside the free-space $\FS(P,\overline{v_{i'}v_i})$ from $(0,\DP(\kappa(i',j'),i',j')$ to $(1,t_{i,j})$. 
\end{proof}

\begin{observation}
 \label{minimal-k}
 For any $i' < i$ if there is a monotone path from $(0,\DP(k,i',j'))$ to $(1,t_{i,j})$ in the free-space $\FS(P,\overline{v_{i'}v_{i}})$ intersecting $\mathit{hor}_j$, then there is also a monotone path from $(0,\DP(\kappa(i',j'),i',j'))$ to $(1,t_{i,j})$ in the free-space $\FS(P,\overline{v_{i'}v_{i}})$ intersecting $\mathit{hor}_j$. 
\end{observation}

\begin{proof}
 
  \begin{figure}
  \centering
 \label{consider-only-kappa(i',j')}
 \begin{subfigure}{0.4\textwidth}

\begin{tikzpicture}

\draw[blue, very thick, fill = green!5] (0,0) rectangle (2,6);
\node[left = 0pt of {(0,0-0.1)}] {$v_{i'}$};
\node[right = 0pt of {(2,0-0.1)}] {$v_{i}$};


\draw (0,2) -- (2,2);
\draw (0,4) -- (2,4);
\draw (0,6) -- (2,6);




\draw[black, ultra thick] (0-0.1,0.75)--(0+0.1,0.75); 
\node[left = 0pt of {(0-0.1,0.75)}]{$\scriptstyle t_{i',j'}$};

\draw[black, ultra thick] (0-0.1,1.75)--(0+0.1,1.75); 
\node[left = 0pt of {(0-0.1,1.75)}]{$\scriptstyle s_{i',j'}$};

\draw[black, ultra thick] (0-0.1,1.05)--(0+0.1,1.05); 
\node[left = 0pt of {(0-0.1,1.05)}]{$\scriptstyle \DP(k,i',j')$};

\draw[black, ultra thick] (0-0.1,1.45)--(0+0.1,1.45); 
\node[left = 0pt of {(0-0.1,1.45)}]{$\scriptstyle \DP(\kappa(i',j'),i',j') $};

\draw[black, ultra thick] (2-0.1,4.15)--(2+0.1,4.15); 
\node[right = 0pt of {(2+0.1,4.15)}]{$\scriptstyle t_{i,j}$};

\draw[black, ultra thick] (2-0.1,5.15)--(2+0.1,5.15); 
\node[right = 0pt of {(2+0.1,5.15)}]{$\scriptstyle s_{i,j}$};

\node[below = 0pt of {(0.45,2)}]{$z$};

\draw[red, ultra thick] (0.4,2)--(1.4,2);
\draw[red, ultra thick] (0.5,4)--(1.6,4);

\draw[black, ultra thick] (0,1.05) -- (0.4, 2) -- (0.5,4) -- (2,4.15);

\end{tikzpicture}
 \end{subfigure}
 \begin{subfigure}{0.4\textwidth}

\begin{tikzpicture}

\draw[blue, very thick, fill = green!5] (0,0) rectangle (2,6);
\node[left = 0pt of {(0,0-0.1)}] {$v_{i'}$};
\node[right = 0pt of {(2,0-0.1)}] {$v_{i}$};


\draw (0,2) -- (2,2);
\draw (0,4) -- (2,4);
\draw (0,6) -- (2,6);




\draw[black, ultra thick] (0-0.1,0.75)--(0+0.1,0.75); 
\node[left = 0pt of {(0-0.1,0.75)}]{$\scriptstyle t_{i',j'}$};

\draw[black, ultra thick] (0-0.1,1.75)--(0+0.1,1.75); 
\node[left = 0pt of {(0-0.1,1.75)}]{$\scriptstyle s_{i',j'}$};

\draw[black, ultra thick] (0-0.1,1.05)--(0+0.1,1.05); 
\node[left = 0pt of {(0-0.1,1.05)}]{$\scriptstyle \DP(k,i',j')$};

\draw[black, ultra thick] (0-0.1,1.45)--(0+0.1,1.45); 
\node[left = 0pt of {(0-0.1,1.45)}]{$\scriptstyle \DP(\kappa(i',j'),i',j')$};

\draw[black, ultra thick] (2-0.1,4.15)--(2+0.1,4.15); 
\node[right = 0pt of {(2+0.1,4.15)}]{$\scriptstyle t_{i,j}$};

\draw[black, ultra thick] (2-0.1,5.15)--(2+0.1,5.15); 
\node[right = 0pt of {(2+0.1,5.15)}]{$\scriptstyle s_{i,j}$};

\node[below = 0pt of {(0.45,2)}]{$z$};

\draw[red, ultra thick] (0.4,2)--(1.4,2);
\draw[red, ultra thick] (0.5,4)--(1.6,4);

\draw[black, ultra thick] (0,1.45) -- (0.4, 2) -- (0.5,4) -- (2,4.15);

\end{tikzpicture}
 \end{subfigure}
 \caption{Illustration of the proof of Observation \ref{minimal-k}. For any $i' < i$, $j' < j$ and any $k$, there is a monotone path from $(0,\DP(k,i',j'))$ to $(1,t_{i,j})$ in $\FS(P,\overline{v_{i'}v_{i}})$ (left) that intersects $\hor_j$ at $z$. Then there is a monotone path from $(0,\DP(\kappa(i',j'),i',j'))$ to $(1,t_{i,j})$ in $\FS(P,\overline{v_{i'}v_{i}})$ (right) by walking from $(0,\DP(\kappa(i',j'),i',j')$ to $z$ and then following the existing monotone path from $z$ to $(1,t_{i,j})$.}
\end{figure}
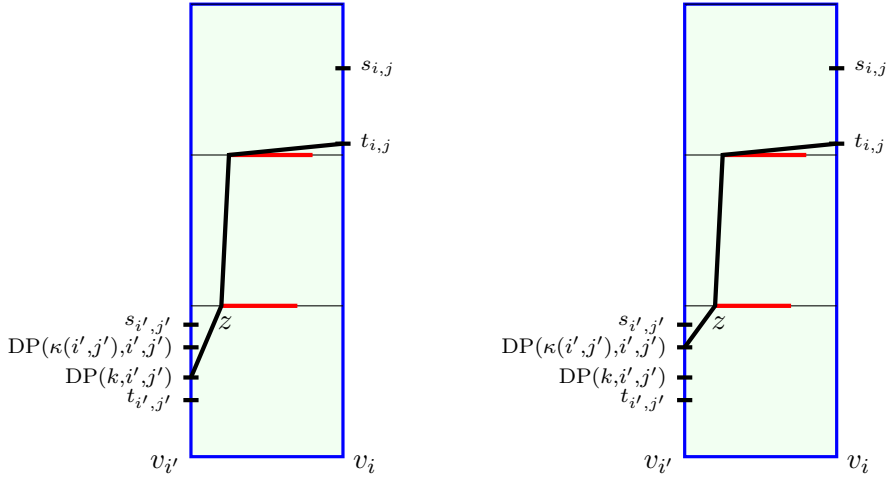
 
 This is obvious by inspecting the free-space $\FS(P,\overline{v_{i'}v_i})$ as follows. Since the monotone path intersects $\mathit{hor}_j$, we have $j' < j$. Observe that both $\DP(k,i',j')$ and $\DP(\kappa(i',j'),i',j')$ lie in the interval $[t_{i',j'},s_{i',j'}]$. Also let $z$ be the point at which the monotone path intersects $\mathit{hor}_{j'+1}$. Then there is a monotone path in $\FS(P,\overline{v_{i'},v_{i}})$ from $z$ to $(1,t_{i,j})$. Since $\mathit{fbox}_{j'}$ is convex (By Fact \ref{convexity}) the line segment joining $(0,\DP(\kappa(i',j'),i',j'))$ and $z$ is contained in $\mathit{fbox}_{j'}$. Therefore there is a monotone path from $(0,\DP(\kappa(i',j'),i',j'))$ to $(1,t_{i,j})$ by walking from $(0,\DP(\kappa(i',j'),i',j'))$ to $z$ and then follow the monotone path from $z$ to $(1,t_{i,j})$.   
\end{proof}

Observations \ref{monotonicity-reachability} and \ref{minimal-k} imply that $\kappa(i,j,i')$ is the minimal value of  $1 + \kappa(i',j')$ over all $j' < j$ such that there exist $x_{j'+1} \leq x_{j'+2} \leq \ldots \leq x_j$ with every $x_k \in [a_k,b_k]$ for all $j' < k \leq j$.

Note that now we are \say{almost} in an instance of Cell Reachability problem where the passage $p_j$ corresponds to the free space on $\mathit{hor}_j$ and each $\lambda_j = 1 + \kappa(i',j)$. The only problem is that the free space on some $\hor_j$ could be empty (while in Cell Reachability section we never had empty passages). However if the free space on any $\hor_j$ is empty then there exists no monotone path in the free-space $\FS(P,\overline{v_{i'}v_i})$ from any any point below $\hor_j$ to any point above $\hor_j$. Thus we can split the instance into two disjoint instances of Cell Reachability. 
Thus for any fixed $i'$ we can determine $\kappa(i,j,i')$ in $\mathcal{O}(n)$ time and therefore we can implement $\kappa_2$-subroutine$(i)$ for any $i \in [n]$ in $T(n) = \mathcal{O}(n^2)$.

\section{Conditional Lower Bound for Curve Simplification}
\label{sec:lowerbound}

In this section we show that an $\mathcal{O}(n^{3-\varepsilon} \textup{poly}(d))$ time algorithm for Global-Fr\'echet, Local-Fr\'echet or Local-Hausdorff simplification over $(\mathbb{R}^d,\norm{}{p})$ for any $p \in [1, \infty)$, $p \neq 2$, would yield an $\mathcal{O}(n^{3-\varepsilon} \mathit{poly}(d))$ algorithm for $\FOV$. 

\subsection{Overview of the Reduction}

We first give an overview of the reduction. Consider any instance $(A,B,C)$ of $\FOV$  where $A$,$B$,$C \subseteq \left\{0,1\right\}^d$ have size $n$. We write $A = \left\{a_1,a_2, \ldots a_n \right\}$, $B = \left\{b_1,b_2, \ldots b_n \right\}$ and $C = \left\{c_1,c_2, \ldots c_n \right\}$. We will construct efficiently a total of $3n+1$ points in $\mathbb{R}^{D}$ with $D \in \mathcal{O}(d)$ namely the sets of points $\tilde{A} = \left\{\tilde{a}_1,\tilde{a}_2, \ldots \tilde{a}_n \right\}$, $\tilde{B} = \left\{\tilde{b}_1,\tilde{b}_2, \ldots \tilde{b}_n \right\}$ and $\tilde{C} = \left\{\tilde{c}_1,\tilde{c}_2, \ldots \tilde{c}_n \right\}$  and one more point $s$. We also determine $\delta \geq 0$ such that the following properties are satisfied.  

\begin{enumerate}
    
    \item[($\mathbf{P}_1$)] For any $\tilde{a} \in \tilde{A}, \tilde{b} \in \tilde{B}, \tilde{c} \in \tilde{C}$, there is a point $x$ on the line segment $\overline{\tilde{a}\tilde{b}}$ with $\norm{x-\tilde{c}}{p} \leq \delta$ if and only if $\norm{\frac{\tilde{a}+\tilde{b}}{2}-\tilde{c}}{p} \leq \delta$. 
    
    \item[($\mathbf{P}_2$)] For any $\tilde{a} \in \tilde{A}, \tilde{b} \in \tilde{B}, \tilde{c} \in \tilde{C}$, we have $\norm{\frac{\tilde{a}+\tilde{b}}{2}-\tilde{c}}{p} \leq \delta$ if and only if $\sum \limits_{\ell \in [d]} a[\ell] \cdot b[\ell] \cdot c[\ell] \neq 0$.   
    
    \item[($\mathbf{P}_3$)] $\norm{x-y}{p} \leq \delta$ holds for all $x,y \in \tilde{A}$, and for all $x,y \in \tilde{B}$ and for all $x,y \in \tilde{C}$.
    
    \item[($\mathbf{P}_4$)] For any $y_1,y_2 \in \left\{s\right\} \cup \tilde{B} \cup \tilde{C}$ and any point $x$ on the line segment $\overline{y_1y_2}$ we have $\norm{x-\tilde{a}}{p} > \delta$ for all $\tilde{a} \in \tilde{A}$. 
    
    \item[($\mathbf{P}_5$)] For any $y_1,y_2 \in \left\{s\right\} \cup \tilde{A} \cup \tilde{C}$ and any point $x$ on the line segment $\overline{y_1y_2}$ we have $\norm{x-\tilde{b}}{p} > \delta$ for all $\tilde{b} \in \tilde{B}$. 
    
    \item[($\mathbf{P}_6$)] For any $y \in \tilde{B} \cup \tilde{A}$ and any point $x$ on the line segment $\overline{sy}$ we have $\norm{x-\tilde{c}}{p} > \delta$ for all $\tilde{c} \in \tilde{C}$.
\end{enumerate}

We postpone the exact construction of these points. Our hard instance for curve simplification will be $Q = \langle s, \tilde{a}_1,\tilde{a}_2, \ldots ,\tilde{a}_n,$ $\tilde{c}_1, \tilde{c}_2, \ldots ,\tilde{c}_n, \tilde{b}_1, \tilde{b}_2, \ldots ,\tilde{b}_n, s \rangle$.

\begin{lemma}
\label{equivalence-allmeasures}
 Let $\hat{Q} = \langle s, \tilde{a}_i, \tilde{b}_j, s \rangle$ for some $\tilde{a}_i \in \tilde{A}$ and $\tilde{b}_j \in \tilde{B}$. If $\norm{\frac{\tilde{a}_i+\tilde{b}_j}{2}-\tilde{c}}{p} \leq \delta$ for all $\tilde{c} \in \tilde{C}$ then the Local-Frechet distance between $Q$ and $\hat{Q}$ is at most $\delta$.
\end{lemma}

\begin{proof}
Both $Q$ and $\hat{Q}$ have the same starting point $s$. By property $\mathbf{P}_1$ we have $\norm{\tilde{a}-\tilde{a}_i}{p} \leq \delta$ for all $\tilde{a} \in \tilde{A}$, and $\norm{\tilde{b}-\tilde{b}_j}{p} \leq \delta$ for all $\tilde{b} \in \tilde{B}$. Thus it follows that $\delta_F(\langle s, \tilde{a}_1, \ldots , \tilde{a}_i \rangle, \overline{s\tilde{a}_i}) \leq \delta$ and $\delta_F(\langle \tilde{b}_j, \ldots , \tilde{b}_n, s \rangle, \overline{\tilde{b}_js}) \leq \delta$. It remains to show that $\delta_F(Q_{ij}, \overline{\tilde{a}_i\tilde{b}_j}) \leq \delta$ where $Q_{ij} = \langle \tilde{a}_i, \ldots ,\tilde{a}_n,\tilde{c}_1, \ldots ,\tilde{c}_n,\tilde{b}_1, \ldots ,\tilde{b}_j \rangle$. To this end first note that both polylines $Q_{ij}$ and $\overline{\tilde{a}_i\tilde{b}_j}$ have the same endpoints. We now outline monotone walks on both $Q_{ij}$ and $\overline{\tilde{a}_i\tilde{b}_j}$.
\begin{enumerate}[label={(\arabic*)}]
 \item Walk on $Q_{ij}$ from $\tilde{a}_i$ to $\tilde{a}_n$ and remain at $\tilde{a}_i$ on $\overline{\tilde{a}_i\tilde{b}_j}$. 
 \item Walk uniformly on both polylines, up to $\frac{\tilde{a}_i+\tilde{b}_j}{2}$ on $\overline{\tilde{a}_i\tilde{b}_j}$ and up to $\tilde{c}_1$ on $Q_{ij}$.
 \item Walk on $Q_{ij}$ from $\tilde{c}_1$ to $\tilde{c}_n$ and remain at $\frac{\tilde{a}_i+\tilde{b}_j}{2}$ on  $\overline{\tilde{a}_i\tilde{b}_j}$.
 \item Walk uniformly on both curves up to $\tilde{b}_j$ on $\overline{\tilde{a}_i\tilde{b}_j}$ and up to $\tilde{b}_1$ on $Q_{ij}$. 
 \item Walk on $Q_{ij}$ until $\tilde{b}_j$ and remain at $\tilde{b}_j$ on $\overline{\tilde{a}_1\tilde{b}_j}$.
\end{enumerate}
 We now argue that we always stay within distance $\delta$ throughout the walks. For (1) and (5) this follows due to property $\mathbf{P}_1$. For (2) and (4) it follows due to the fact we always remain within distance $\delta$ while walking with uniform speed on two line segments, as long as their startpoints and their endpoints are within distance $\delta$. By the assumption $\norm{\frac{\tilde{a}_i+\tilde{b}_j}{2}-\tilde{c}}{p} \leq \delta$ for all $\tilde{c} \in \tilde{C}$, we always stay within distance $\delta$ also for (3). 
\end{proof}

Observe that property $\mathbf{P_3}$ implies that there is a simplification of size five namely $\hat{Q} = \langle s, \tilde{a}, \tilde{c} ,\tilde{b}, s \rangle$  for any $\tilde{a} \in \tilde{A}$, $\tilde{b} \in \tilde{B}$, and $\tilde{c} \in \tilde{C}$, such that the distance between $\hat{Q}$ and $Q$ is at most $\delta$ under Local-Fr\'echet, Global-Fr\'echet and Local-Hausdorff distance. We now show that a smaller simplification is only possible if there exist $a \in A$, $b \in B$ such that for all $c \in C$ we have $\sum \limits_{\ell \in [d]}a[\ell] \cdot b[\ell] \cdot c[\ell] \neq 0$. 

\begin{lemma}
\label{equivalence-allmeasures2}
 Let $\hat{Q}$ be a simplification of the polyline $Q$ of size 4. Then the following statements are equivalent
 \begin{enumerate}[label={(\arabic*)}]
     \item The Global-Fr\'echet distance between $Q$ and $\hat{Q}$ is at most  $\delta$.
     \item The Local-Fr\'echet distance between $Q$ and $\hat{Q}$  is at most  $\delta$.
     \item The Local-Hausdorff distance between $Q$ and $\hat{Q}$  is at most  $\delta$.
     \item There exist some $\tilde{a} \in \tilde{A}$, $\tilde{b} \in \tilde{B}$, such that $\hat{Q} = \langle s,\tilde{a}, \tilde{b}, s \rangle$ and $\norm{\frac{\tilde{a}+\tilde{b}}{2}-\tilde{c}}{p} \leq \delta$ for every $\tilde{c} \in \tilde{C}$.
     \item There exist $a \in A$, $b \in B$ such that for all $c \in C$ we have $\sum \limits_{\ell \in [d]}a[\ell] \cdot b[\ell] \cdot c[\ell] \neq 0$.
 \end{enumerate}
\end{lemma}

\begin{proof}
   We first show that (1), (2) and (3) are equivalent to (4). To this end, we first show that each of (1), (2) and (3) imply (4). Since for any $y_1,y_2 \in s \cup \tilde{B} \cup \tilde{C}$ there is no point on the line segment $\overline{y_1y_2}$ that has distance at most $\delta$ to any $\tilde{a} \in \tilde{A}$ (by property $\mathbf{P}_4$), $\hat{Q}$ must contain at least one point from $\tilde{A}$. A symmetric argument can be made for the fact that $\hat{Q}$ must contain at least one point from $\tilde{B}$ (property $\mathbf{P}_5$). Since the size of $\hat{Q}$ is $4$, we have $\hat{Q} = \langle s, \tilde{a},\tilde{b}, s \rangle$ for some $\tilde{a} \in \tilde{A}$ and $\tilde{b} \in \tilde{B}$. By property $\mathbf{P}_6$ there is no point on the line segments $\overline{s\tilde{a}}$ and $\overline{\tilde{b}s}$ that has distance at most $\delta$ to any $\tilde{c} \in C$. Therefore the Global-Fr\'echet distance or the Local-Fr\'echet distance or the Local-Hausdorff distance between $Q$ and $\hat{Q}$ is at most $\delta$ only if for all $\tilde{c} \in \tilde{C}$ there is a point on the line segment $\overline{\tilde{a}\tilde{b}}$ that has distance at most $\delta$ to $\tilde{c}$. By property $\mathbf{P}_1$, this implies that $\norm{\frac{\tilde{a}+\tilde{b}}{2} - \tilde{c}}{p} \leq \delta$ for all $\tilde{c} \in \tilde{C}$.
   
 Now we show that (4) implies (1), (2) and (3). First observe that (2) implies (1) and (3), since the Local-Fr\'echet distance between a curve and its simplification is at least the Global-Fr\'echet distance and at least the Local-Hausdorff distance between the same. Thus, it suffices to show that (4) implies (2). This directly follows from Lemma \ref{equivalence-allmeasures}. Finally, (4) and (5) are equivalent due to property $\mathbf{P}_2$.
\end{proof}

Assuming that we can construct $Q$ and determine $\delta$ in $\mathcal{O}(nd)$ time, the above lemma directly yields the following theorem,

\begin{theorem}
 \label{lowerboundtheorem}
 For any $\varepsilon > 0$, there is no $\mathcal{O}(n^{3- \varepsilon} \textup{poly}(d))$ algorithm for Global-Fr\'echet, Local-Fr\'echet  and Local-Hausdorff simplification over $(\mathbb{R}^d,\norm{}{p})$ for any $p \in [1,\infty)$, $p\neq 2$ unless $\FOVH$ fails.
\end{theorem}

\begin{proof}
 The curve $Q$ can be constructed and $\delta$ can be determined in $\mathcal{O}(nd)$ from any instance $A,B,C$ of $\FOV$. Henceforth, by Lemma \ref{equivalence-allmeasures2} the simplification problem is equivalent to $\FOV$. Thus any $\mathcal{O}(n^{3-\varepsilon} \textup{poly}(d))$ algorithm for the curve simplification problem will yield an $\mathcal{O}(n^{3-\varepsilon} \textup{poly}(d))$ algorithm for $\FOV$ as well.  
\end{proof}

It remains to construct the point $s$ and the sets $\tilde{A}$, $\tilde{B}$ and $\tilde{C}$ and determine $\delta$ in $\mathcal{O}(nd)$. We first introduce some notations. For vectors $x$ and $y$ and $\alpha \in [-\frac{1}{2},\frac{1}{2}]$, we define $P_{xy}(\alpha)$ as $(\frac{1}{2} - \alpha)x + (\frac{1}{2}+\alpha)y$.
Moreover let $u_i \in \mathbb{R}^{d}$. We write $v = \big[ u_1u_2 \ldots u_m \big]$ for the vector $v \in \mathbb{R}^{md}$ with $v[(j-1)d + k] = u_j[k]$ for any $j \in [m]$ and $k \in [d]$.
\begin{observation}
\label{vector-decomposition}
Let $u_1,u_2, \ldots ,u_m \in \mathbb{R}^d$ and $v = \big[ u_1u_2 \ldots u_m \big]$. Then we have $\norm{v}{p}^p = \sum \limits_{i \in [m]} \norm{u_i}{p}^p$.
\end{observation}

\subsection{Cordinate gadgets}
In this section our aim is to construct points $\as{i}$, $\bs{i}$, $\cs{i}$ for $i \in \left\{0,1\right\}$ such that the distance $\norm{\cs{i} - P_{\as{j}\bs{k}}(0)}{p}$ only depends on whether the bits $i,j,k \in \left\{0,1\right\}$ seen as cordinates of vectors are orthogonal. In other words the points $\as{i}$, $\bs{i}$, $\cs{i}$ form a cordinate gadget. Formally we will prove the following lemma,

\begin{lemma}
 \label{main-lemma-cordinate}
 For any $p \neq 2$
\begin{align*}
 \norm{\cs{i} -P_{\as{j}\bs{k}}(0)}{p}^p = \left\{
                                                          \begin{array}{ll}
                                                                \beta_1    & \text{ if }  i=1,j=1,k=1\\
                                                                \beta_2   & \text{ otherwise }  
                                                                \end{array} 
                                                        \right.
 \end{align*}
 where $\beta_1 < \beta_2$.
\end{lemma}

In Section \ref{vectorgadgets} we will use this lemma to construct the final point sets $\tilde{A}$, $\tilde{B}$ and $\tilde{C}$.

Let $\theta_1$, $\theta_2$, $\theta_3$, $\theta_4$ and $\theta_5$ be positive constants. We construct the points $\as{1}$,$\bs{1}$,$\cs{1}$ and $\as{0}$,$\bs{0}$,$\cs{0}$ in $\mathbb{R}^{9}$ as follows,

\begin{align*}
    \as{0} &= \big[ &-\theta_1  &, & 0         & , & -\theta_2    &,& 0         &,& \theta_3  &, & 2\theta_3     & , & \theta_4    &,&  -2\theta_4 &,&   0  \big]\\
    \as{1} &= \big[ &\theta_1   &, & 2\theta_1 & , & \theta_2     &,&-2\theta_2 &,& 0         &, &-\theta_3      & , &-\theta_4    & ,&  0         & , & 0 \big]\\
    \bs{0} &= \big[ &-\theta_1  &, & 0         & , & \theta_2     &,& 2\theta_2 &,& \theta_3  &, &-2\theta_3     & , &-\theta_4    & ,&  0         & , & 0   \big]\\
    \bs{1} &= \big[ &\theta_1   &, &-2\theta_1 & , & -\theta_2    &,& 0         &,& 0         &, & -\theta_3     & , & \theta_4    &, &  2\theta_4 & , & 0 \big]\\
    \cs{0} &= \big[ &0          &, & 0         & , & 0            &,& 0         &,& 0         &, & 0             & , & 0           & ,&  0         & , & \theta_5 \big]\\
    \cs{1} &= \big[ &-\theta_1  &, & 0         & , & -\theta_2    &,& 0         &,& -\theta_3 &, & 0             & , & -\theta_4   & , & 0         &, & 0  \big]
\end{align*}

From these points we can compute the points $P_{\as{i}\bs{j}}(0)$ for all $i,j \in \left\{0,1\right\}$.

\begin{align*}
    P_{\as{0}\bs{0}}(0) &= \big[  &-\theta_1 &,   &0         &,    & 0          &, & \theta_2     &, & \theta_3   &, & 0           &,& 0         &,& -\theta_4   &,&0 \big]\\ 
    P_{\as{1}\bs{0}}(0) &= \big[  & 0        &,   &\theta_1  &,    & \theta_2   &,   & 0          &, & 0          &,  & -\theta_3  &,&-\theta_4  &,& 0           &,&0 \big]\\
    P_{\as{1}\bs{1}}(0) &=  \big[ &\theta_1  &,   & 0        &,    & 0          &,   & -\theta_2  &, &-\theta_3   &, & 0           &,&0          &,& \theta_4    &,&0 \big]\\
    P_{\as{0}\bs{1}}(0) &= \big[  & 0        &,   &-\theta_1 &,    &-\theta_2   &,   & 0          &, & 0          &,  &\theta_3    &,&\theta_4   &,& 0           &,&0\big]
\end{align*}

Observe that  $\norm{\cs{0}-P_{\as{i}\bs{j}}(0)}{p}^p = \sum \limits_{r \in [5]} \theta_r^p$ for all $i,j \in \left\{0,1\right\}$. Thus all the points $P_{\as{i}\bs{j}(0)}$ are equidistant from $\cs{0}$ irrespective of the exact values of the $\theta_r$ for $r \in [5]$. Note that when $\theta_r = \theta$ for all $r \in [5]$, then $\norm{\cs{1} - P_{\as{i}\bs{j}}(0)}{p}^p = 4\theta^p + 2^p\theta^p$ for all $i,j \in \left\{0,1\right\}$. Thus all the points $P_{\as{i}\bs{j}}(0)$ are equidistant from $\cs{1}$ when all the $\theta_r$ are the same. We now determine $\theta_r$ for $r \in [5]$ such that all but one point in $\left\{P_{\as{i}\bs{j}}(0)| i,j \in \left\{0,1\right\} \right\}$ are equidistant and far from $\cs{1}$. More precisely,  

\begin{align*}
\norm{\cs{1} -P_{\as{i}\bs{j}}(0)}{p}^p = \left\{
                                                          \begin{array}{ll}
                                                                \beta_1    & \text{ if }  i=1,j=1\\
                                                                \beta_2   & \text{ otherwise }  
                                                                \end{array} 
                                                        \right.
\end{align*}
 and $\beta_1 < \beta_2$.\\
 We first quantify the distances from $\left\{\cs{0},\cs{1}\right\}$ to each of the points in $\left\{P_{\as{j}\bs{k}}(0)| j,k \in \left\{0,1\right\} \right\}$.
 
\begin{lemma}
\label{distance-c1}
 We have 
 \begin{align*}
  \norm{\cs{i} -P_{\as{j}\bs{k}}(0)}{p}^p &= \left\{
                                                          \begin{array}{ll}
                                                                \sum_{r \in [5]} \theta_r^p                  &\text{ if }  i=0\\
                                                                2\theta_2^p + 2^p\theta_3^p + 2\theta_4^p    & \text{ if } i=1, j=0,k=0\\
                                                                2\theta_1^p + 2^p \theta_2^p + 2\theta_3^p  & \text{ if } i=1, j=1,k=0\\
                                                                2\theta_1^p + 2\theta_3^p + 2^p \theta_4^p  & \text{ if } i=1, j=0,k=1\\
                                                                2^p\theta_1^p + 2\theta_2^p + 2\theta_4^p    & \text{ if } i=1,j=1,k=1
                                                                \end{array}
                                                          \right.
\end{align*}
\end{lemma} 

We now set the exact values of $\theta_r$ for $r \in [5]$. We define values depending on $p$. When $1 \leq p < 2$ we set 
\[ \theta_1 = (2^{p-1}-1)^{\frac{1}{p}}, \theta_2 =0, \theta_3 =1,  \theta_4 =0, \theta_5 = 2^{\frac{p-1}{p}} \]

Now we make the following observation,

\begin{observation}
\label{cordinate-1}
When $1\leq p < 2$, then 
\begin{align*}
\norm{\cs{i} -P_{\as{j}\bs{k}}(0)}{p}^p = \left\{
                                                          \begin{array}{ll}
                                                                2^p (2^{p-1}-1) & \text{ if }  i=1, j=1,k=1\\
                                                                2^p   & \text{ otherwise }  
                                                                \end{array} 
                                                        \right.
\end{align*}
\end{observation}

\begin{proof}
  Substituting the values of $\theta_k$ for every $k \in [4]$ in Lemma \ref{distance-c1} we have that 
  \begin{align*}
     \norm{\cs{0} -P_{\as{j}\bs{k}}(0)}{p}^p &= 2^{p-1}-1 + 0^p + 1^p + 0^p + 2^{p-1} &= 2^{p}\\ 
     \norm{\cs{1} -P_{\as{0}\bs{0}}(0)}{p}^p &= 2 \cdot 0^p + 2^p \cdot 1^p + 2 \cdot 0^p  &= 2^p \\
     \norm{\cs{1} -P_{\as{1}\bs{0}}(0)}{p}^p &= 2 \cdot (2^{p-1}-1) + 2^p \cdot 0^p + 2 \cdot 1^p  &= 2^p\\
     \norm{\cs{1} -P_{\as{0}\bs{1}}(0)}{p}^p &= 2 \cdot (2^{p-1}-1) + 2 \cdot 1^p + 2^p \cdot 0^p  &= 2^p\\
     \norm{\cs{1} -P_{\as{1}\bs{1}}(0)}{p}^p &= 2^p \cdot (2^{p-1}-1) + 2 \cdot 0^p + 2^p \cdot 0^p &= 2^p(2^{p-1}-1)
  \end{align*}
\end{proof}

In case $p >2$. Then we set \[\theta_1 = 0, \theta_2 =(2^p-2)^{\frac{1}{p}}, \theta_3 =(2^p-4)^{\frac{1}{p}}, \theta_4 = (2^p-2)^{\frac{1}{p}}, \theta_5 = (2^{2p}-3 \cdot 2^p)^\frac{1}{p} \] We make a similar observation,

\begin{observation}
\label{cordinate-2}
When $p > 2$, then 
\begin{align*}
\norm{\cs{i} -P_{\as{j}\bs{k}}(0)}{p}^p = \left\{
                                                          \begin{array}{ll}
                                                                2^{p+2}-8    & \text{ if }  i=1,j=1, k=1\\
                                                                2^{2p}-8     & \text{ otherwise }  
                                                                \end{array} 
                                                          \right.
\end{align*}
\end{observation}

\begin{proof}
  Substituting the values of $\theta_k$ for every $k \in [4]$ in Lemma \ref{distance-c1} we have that 
  \begin{align*}
     \norm{\cs{0} -P_{\as{j}\bs{k}}(0)}{p}^p &= 0^p + (2^p-2) + (2^p-4) + (2^p-2) + (2^{2p}-3 \cdot 2^p) &= 2^{2p}-8\\ 
     \norm{\cs{1} -P_{\as{0}\bs{0}}(0)}{p}^p &= 2 \cdot (2^p-2) + 2^p \cdot (2^p-4) + 2 \cdot(2^p-2)  &=  2^{2p}-8\\
     \norm{\cs{1} -P_{\as{1}\bs{0}}(0)}{p}^p &= 2\cdot 0^p + 2^p \cdot (2^p-2) + 2 \cdot (2^p-4)  &=  {2^{2p}-8}\\
     \norm{\cs{1} -P_{\as{0}\bs{1}}(0)}{p}^p &= 2\cdot 0^p +  2 \cdot (2^p-4) + 2^p \cdot (2^p-2)  &= 2^{2p}-8\\
     \norm{\cs{1} -P_{\as{1}\bs{1}}(0)}{p}^p &= 2^p \cdot 0^p + 2 \cdot (2^p-2) + 2 \cdot (2^p-2) &= 2^{p+2}-8
  \end{align*}
\end{proof}

Combining Observations \ref{cordinate-1} and \ref{cordinate-2} we arrive at Lemma $\ref{main-lemma-cordinate}$.

\subsection{Vector gadgets}
\label{vectorgadgets}
For every $a \in A$, $b \in B$ and $c\in C$ we introduce vectors $a',b',c'$ and $a''b'',c''$ and then concatenate the respective vectors to form $\tilde{a}$, $\tilde{b}$ and $\tilde{c}$ respectively. Intuitively $a',b',c'$ primarily helps us to ensure properties $\mathbf{P}_1$ and $\mathbf{P}_2$, while $a''b'',c''$ help us ensure the remaining properties.

\subsubsection{The vectors $a'$, $b'$, $c'$, and $s'$}
We construct the vector $s'$ and the vectors $a'$, $b'$ and $c'$  for every $a \in A$, $b \in B$ and $c \in C$ respectively,  in $\mathbb{R}^{9d}$ as follows,

\begin{align}
 a' &= \big[ \as{a[1]}, \as{a[2]}, \ldots \as{a[d]} \big]\\
 b' &= \big[ \bs{b[1]}, \bs{b[2]}, \ldots \bs{b[d]} \big]\\
 c' &= \big[ \cs{c[1]}, \cs{c[2]}, \ldots \cs{c[d]} \big]\\
 s' &= \big[ 0, 0, \ldots ,0 \big]
\end{align}
 
 We also define the sets $A' = \left\{a'\mid a \in A\right\}$, $B' = \left\{b'\mid b \in B\right\}$ and $C' = \left\{c'\mid c \in C\right\}$. We now make a technical observation about the vectors in $A',B'$, and $C'$, that will be useful later. We set $\eta_1 = \max \limits_{i \in [5]} \theta_i$. 
 
\begin{observation}
\label{technical-lowerbound}
 For any $x,y \in A' \cup B' \cup C'$, we have $\norm{x-y}{p} \leq \eta_2$ where $\eta_2 \colon = 36d\eta_1$. 
\end{observation}

\begin{proof}
   Note that the absolute value of every cordinate of the vectors $\as{1},\bs{1},\cs{1}$ and $\as{0},\bs{0},\cs{0}$ is bounded by $2\eta_1$ (Since every cordinate is of the form $ \pm \theta_r$ or $ \pm 2\theta_r$ or $0$). Also  every cordinate of $a'$, $b'$, and $c'$, is a cordinate of one of $\as{1},\bs{1},\cs{1},\as{0},\bs{0}$ and $\cs{0}$. Therefore for any $x,y \in A' \cup B' \cup C'$ we have $\max \limits_{\ell \in [9d]} \abs{x[\ell]-y[\ell]} \leq 4\eta_1$. Hence we have $\norm{x-y}{p} \leq \sum_{\ell \in [9d]} \abs{x[\ell] -y[\ell]} \leq 9d \cdot 4\eta_1 = 36d\eta_1 = \eta_2$.
  \end{proof}

 Note that $a \in A$, $b \in B$ and $c \in C$ are non orthogonal if and only if $\#_{111}^{c,a,b} > 0$. The following Lemma shows a connection between non-orthogonality and small distance $\norm{c' - P_{a'b'}(0)}{p}$.  
  
 \begin{lemma}
  \label{vector-gadget-1}
   For any $a \in A$, $b\in B$ and $c \in C$ we have $\norm{c' - P_{a'b'}(0)}{p}^p = d\beta_2 -(\beta_2-\beta_1)\#_{111}^{c,a,b}$.
 \end{lemma}
 
 \begin{proof}
  By Lemma \ref{main-lemma-cordinate}, for any $\alpha \in [-\frac{1}{2},\frac{1}{2}]$
  \begin{align*}
     \norm{\cs{c[\ell]} - P_{\as{a[\ell]}\bs{b[\ell]}}(0)}{p}^p = \left\{
                                                                            \begin{array}{ll}
                                                                            \beta_1    & \text{ if }  c[\ell]=a[\ell]=b[\ell]=1\\
                                                                            \beta_2   & \text{ otherwise }  
                                                                            \end{array} 
                                                                      \right.  
  \end{align*} 
  By Observation $\ref{vector-decomposition}$ we have
  \begin{align*}
      \norm{c' - P_{a'b'}(0)}{p}^p &= \sum \limits_{\ell \in [d]} \norm{\cs{c[\ell]} - P_{\as{a[\ell]}\bs{b[\ell]}}(0)}{p}^p\\
                                                  &= \beta_2(d - \#_{111}^{c,a,b}) + \beta_1\#_{111}^{c,a,b}\\
                                                  &= d\beta_2 - (\beta_2-\beta_1)\#_{111}^{c,a,b}. \qedhere
  \end{align*}
 \end{proof}

\subsubsection{The vectors $a''$,$b''$, $c''$, and $s''$} 
We construct the vector $s''$ and the vectors $a''$,$b''$, and $c''$  for every $a \in A$,$b \in B$, and $c \in C$, respectively  in $\mathbb{R}^{3}$ as follows,

\begin{align*}
    a'' &= \big[\gamma_1,0,0 \big]\\
    b'' &= \big[\gamma_1,\gamma_2,0 \big]\\
    c'' &= \big[0,\frac{\gamma_2}{2},0\big]\\
    s'' &= \big[0,\frac{\gamma_2}{2},\gamma_2]
\end{align*} 
where $\gamma_1$, $\gamma_2$ are positive constants. We are now ready to define the final points of our construction, $s$ and $\tilde{a}$, $\tilde{b}$ and $\tilde{c}$ for any $a \in A$, $b \in B$ and $c \in C$ respectively.

\begin{align*}
    \tilde{a} &= \big[a',a''\big]\\
    \tilde{b} &= \big[b',b''\big]\\
    \tilde{c} &= \big[c', c''\big]\\
    s&= \big[ s',s'' \big]
\end{align*} 

 We set \[ \gamma_1 = \eta_2, \delta = (\gamma_1^p + d\beta_2 - (\beta_2-\beta_1))^{\frac{1}{p}}, \gamma_2 = \max \Bigg( 4 \delta, \eta_2 \Bigg( 1 + \frac{(\gamma_1^p + d\beta_2)^\frac{1}{p}}{(\gamma_1^p + d\beta_2)^\frac{1}{p} - \delta}\Bigg) \Bigg) \] 
 
 Note that we have constructed the point sets $\tilde{A}$, $\tilde{B}$, $\tilde{C}$, and the point $s$ and determined $\delta$ in total time $\mathcal{O}(nd)$. Therefore now it suffices to show that our point set and $\delta$ satisfy the properties $\mathbf{P}_1$, $\mathbf{P}_2$, $\mathbf{P}_3$, $\mathbf{P}_4$, $\mathbf{P}_5$, and $\mathbf{P}_6$. To this end we first show  how the distance $\norm{\tilde{c} - P_{\tilde{a}\tilde{b}}(\alpha)}{p}$ is related with $\#_{111}^{c,a,b}$ (the non orthogonality of the vectors $a$,$b$, and $c$) by the following lemma.

\begin{lemma}
\label{main-lemma}
 For any $a \in A$, $b \in B$ and $c \in C$ we have,
 \begin{itemize}
     \item $\norm{\tilde{c} - P_{\tilde{a}\tilde{b}}(0)}{p}^p = \gamma_1^p + \beta_2d - (\beta_2-\beta_1)\#_{111}^{c,a,b}$.
     \item If $\#_{111}^{c,a,b} =0$ then $\norm{\tilde{c} - P_{\tilde{a}\tilde{b}}(\alpha)}{p}^p > \delta$ for all $\alpha \in [-\frac{1}{2},\frac{1}{2}]$.
 \end{itemize}
\end{lemma}

\begin{proof}
 Note that
\begin{align*}
    \tilde{c} - P_{\tilde{a}\tilde{b}}(\alpha) &= \big[c' - P_{a'b'}(\alpha), -\gamma_1, -\gamma_2\alpha, 0 \big]\\
                                               &= \big[c' - P_{a'b'}(0), -\gamma_1, -\gamma_2\alpha, 0 \big] - \big[P_{a'b'}(\alpha) - P_{a'b'}(0), 0, 0, 0 \big]
\end{align*}
 Thus substituting  $\alpha$ as 0,
 \begin{align*}
    \norm{\tilde{c} - P_{\tilde{a}\tilde{b}}(0)}{p}^p &= \norm{[c'- P_{a'b'}(0),-\gamma_1]}{p}^p\\
                                                    &= \gamma_1^p + \norm{c'- P_{a'b'}(0)}{p}^p\\
                                                    &= \gamma_1^p + d\beta_2 - (\beta_2-\beta_1)\#_{111}^{c,a,b} & (\text{by Lemma \ref{vector-gadget-1}})
 \end{align*}
 Furthermore, by reverse triangle inequality we have
\begin{align*}
    \norm{\tilde{c} - P_{\tilde{a}\tilde{b}}(\alpha)}{p} &\geq \norm{\big[c' - P_{a'b'}(0), -\gamma_1, -\gamma_2\alpha, 0 \big]}{p} - \norm{\big[P_{a'b'}(\alpha) - P_{a'b'}(0), 0, 0, 0 \big]}{p}\\
                                                         & = \norm{\big[c' - P_{a'b'}(0), -\gamma_1, -\gamma_2\alpha \big]}{p} - \norm{\big[P_{a'b'}(\alpha) - P_{a'b'}(0) \big]}{p}.
                                                         \end{align*}
We bound the two summands on the right hand side. Note that $\norm{\big[c' - P_{a'b'}(0), -\gamma_1, -\gamma_2\alpha \big]}{p} \geq \max((\gamma_1^p+ d\beta_2 - (\beta_2-\beta_1)\#_{111}^{c,a,b})^{\frac{1}{p}},\abs{\alpha} \gamma_2)$. We also have $\norm{\big[P_{a'b'}(\alpha) - P_{a'b'}(0) \big]}{p} = \abs{\alpha} \norm{b-a}{p} \leq \abs{\alpha}\eta_2$ (by Observation \ref{technical-lowerbound}). Therefore when $\#_{111}^{c,a,b} = 0$, for any $\alpha \in [-\frac{1}{2},\frac{1}{2}]$ we have,
\begin{align*}
    \norm{\tilde{c} - P_{\tilde{a}\tilde{b}}(\alpha)}{p} &\geq \max((\gamma_1^p+ d\beta_2)^{\frac{1}{p}},\abs{\alpha} \gamma_2) - \abs{\alpha} \eta_2
\end{align*}

Now we consider two cases. If $\abs{\alpha} < \frac{1}{\eta_2}((\gamma_1^p + d\beta_2)^\frac{1}{p} - \delta)$, then 
\begin{align*}
    \norm{\tilde{c} - P_{\tilde{a}\tilde{b}}(\alpha)}{p} &> (\gamma_1^p+ d\beta_2)^\frac{1}{p} - ((\gamma_1^p + d\beta_2)^\frac{1}{p} - \delta)\\
                                                         &= \delta
\end{align*}

Similarly if $\abs{\alpha} \geq \frac{1}{\eta_2}((\gamma_1^p + d\beta_2)^\frac{1}{p} - \delta)$, we have 
\begin{align*}
    \norm{\tilde{c} - P_{\tilde{a}\tilde{b}}(\alpha)}{p} &\geq \abs{\alpha} \gamma_2 - \abs{\alpha} \eta_2\\
                                                         &= \abs{\alpha}(\gamma_2-\eta_2)\\
                                                         & \geq \frac{1}{\eta_2}(\gamma_1^p + d\beta_2)^\frac{1}{p} - \delta) \cdot \eta_2 \bigg(\frac{(\gamma_1^p + d\beta_2)^\frac{1}{p}}{(\gamma_1^p 															+ d\beta_2)^\frac{1}{p} - \delta} \bigg) & (\text{ substituting } \gamma_2 \text{ and } \alpha)\\
                                                         &= (\gamma_1^p + d\beta_2)^\frac{1}{p}\\
                                                         &> \delta.
\end{align*}
Combining the two cases, we arrive at the second result of the lemma.
\end{proof}

We now verify properties $\mathbf{P}_1$, $\mathbf{P}_2$, $\mathbf{P}_3$, $\mathbf{P}_4$, $\mathbf{P}_5$, and $\mathbf{P}_6$.

\begin{lemma}[$\mathbf{P}_2$]
\label{Property-1}
  For any $a\in A$, $b\in B$ and $c \in C$ we have $\norm{\tilde{c} - P_{\tilde{a}\tilde{b}}(0)}{p} \leq \delta$ if and only if $\#_{111}^{c,a,b} \geq 1$ or equivalently when $\sum _ {\ell \in [d]} a[\ell] \cdot b[\ell] \cdot c[\ell] \neq 0$.
\end{lemma}

\begin{proof}
 By Lemma \ref{main-lemma} we have that $\norm{\tilde{c} - P_{\tilde{a}\tilde{b}}(0)}{p} = (\gamma_1^p + d\beta_2 - (\beta_2-\beta_1)\#_{111}^{c,a,b})^{\frac{1}{p}}$. Therefore, if $\#_{111}^{c,a,b} \geq 1$ then $\norm{\tilde{c} - P_{\tilde{a}\tilde{b}}(0)}{p}\leq \delta$. Conversely if $\#_{111}^{c,a,b} = 0$  then $\norm{\tilde{c} - P_{\tilde{a}\tilde{b}}(0)}{p} = (\gamma_1^p + d\beta_2)^{\frac{1}{p}}  > (\gamma_1^p + d\beta_2 - (\beta_2-\beta_1))^{\frac{1}{p}} =  \delta$.
\end{proof}

\begin{lemma}[$\mathbf{P}_1$]
\label{Property-2}
 For any $a\in A$, $b\in B$ and $c \in C$ we have $\norm{\tilde{c} - P_{\tilde{a}\tilde{b}}(\alpha)}{p} \leq \delta$ for any $\alpha \in [-\frac{1}{2},\frac{1}{2}]$, if and only if $\norm{\tilde{c} - P_{\tilde{a}\tilde{b}}(0)}{p} \leq \delta$.  
\end{lemma}

\begin{proof}
 The \say{if} statement is trivial as $\norm{\tilde{c} - P_{\tilde{a}\tilde{b}}(\alpha)}{p} \leq \delta$ for $\alpha=0$. For the \say{only if} case, since $\norm{\tilde{c} - P_{\tilde{a}\tilde{b}}(0)}{p} > \delta$, from Lemma \ref{Property-1} it follows that $\#_{111}^{c,a,b} = 0$. By Lemma $\ref{main-lemma}$ we obtain $\norm{\tilde{c} - P_{\tilde{a}\tilde{b}}(\alpha)}{p}^p > \delta$ for all $\alpha \in [-\frac{1}{2},\frac{1}{2}]$. Therefore there exists no $\alpha \in [-\frac{1}{2},\frac{1}{2}]$ such that $\norm{\tilde{c} - P_{\tilde{a}\tilde{b}}(\alpha)}{p} \leq \delta$.   
\end{proof}

\begin{lemma}[$\mathbf{P}_3$]
 We have $\norm{x-y}{p} \leq \delta$ for all $x,y \in \tilde{A}$, and for all $x,y \in \tilde{B}$ and for all $x,y \in \tilde{C}$.
\end{lemma}

\begin{proof}
  We prove the case of $x,y \in \tilde{A}$; the other cases are analogous. Consider any $\tilde{a}_1, \tilde{a}_2 \in \tilde{A}$. Note that $\norm{\tilde{a}_1-\tilde{a}_2}{p} = \norm{a'_1 - a'_2}{p}$. By Observation \ref{technical-lowerbound}, we have $\norm{a'_1 - a'_2}{p} \leq \eta_2 \leq \gamma_1 \leq \delta$.
\end{proof}

We now prove properties $\mathbf{P}_4$, $\mathbf{P}_5$ and $\mathbf{P}_6$.  

\begin{lemma}[$\mathbf{P}_4$,$\mathbf{P}_5$, and $\mathbf{P}_6$] \label{lastlemma}
 For any $a \in A$, $b \in B$ and $c \in C$ and $\alpha \in [-\frac{1}{2},\frac{1}{2}]$ the following properties hold.
 \begin{enumerate}
  \item For any $y_1,y_2 \in \left\{s\right\} \cup \tilde{B} \cup \tilde{C}$, we have $\norm{\tilde{a}-P_{y_1y_2}(\alpha)}{p} > \delta$ for all $\tilde{a} \in \tilde{A}$.
  \item For any $y_1,y_2 \in \left\{s\right\} \cup \tilde{A} \cup \tilde{C}$, we have $\norm{\tilde{b}-P_{y_1y_2}(\alpha)}{p} > \delta$ for all $\tilde{b} \in \tilde{B}$.   
  \item For any $y \in \tilde{A} \cup \tilde{B}$, we have $\norm{\tilde{c}-P_{sy}(\alpha)}{p} > \delta$ for all $\tilde{c} \in \tilde{C}$.
 \end{enumerate}
\end{lemma}

\begin{proof}
  Since we set $\gamma_2 $ to at least $4\delta$, we have $\frac{\gamma_2}{2} > \delta$.  We first prove (1). For any $y_1,y_2 \in \left\{s\right\} \cup \tilde{B} \cup \tilde{C}$ we have $y_1[9d+2] \geq \frac{\gamma_2}{2}$ and $y_2[9d+2] \geq \frac{\gamma_2}{2}$. Therefore for any $\alpha \in [-\frac{1}{2},\frac{1}{2}]$ we have $P_{y_1y_2}(\alpha)[9d+2] \geq \frac{\gamma_2}{2}$.  For any $\tilde{a} \in \tilde{A}$ we have $\tilde{a}[9d+2] = 0$. Hence we obtain $\norm{\tilde{a}-P_{\tilde{y_1}\tilde{y_2}}(\alpha)}{p} \geq \abs{\tilde{a}[9d+2] - P_{y_1y_2}(\alpha)[9d+2]} \geq \frac{\gamma_2}{2} > \delta$.
  
  We now make a symmetric argument for (2). For any $y_1,y_2 \in \left\{s\right\} \cup \tilde{A} \cup \tilde{C}$ we have $y_1[9d+2] \leq \frac{\gamma_2}{2}$ and $y_2[9d+2] \leq \frac{\gamma_2}{2}$.  Therefore for any $\alpha \in [-\frac{1}{2},\frac{1}{2}]$ we have $P_{y_1y_2}(\alpha)[9d+2] \leq \frac{\gamma_2}{2}$. For any $\tilde{b} \in \tilde{B}$ we have $\tilde{b}[9d+2] = \gamma$. Like earlier we obtain $\norm{\tilde{b}-P_{\tilde{y_1}\tilde{y_2}}(\alpha)}{p} \geq \abs{\tilde{b}[9d+2] - P_{y_1y_2}(\alpha)[9d+2]} \geq \frac{\gamma_2}{2} > \delta$.
  
  We now show (3). For this we state a simple observation.
 
  \begin{observation}
   \label{largedistanceC}
    For any $\alpha \in [-\frac{1}{2},\frac{1}{2}]$ we have,
    \begin{itemize}
        \item $\norm{c''-P_{b''s''}(\alpha)}{p} \geq \frac{\gamma_2}{3} >\delta$.
        \item $\norm{c''-P_{a''s''}(\alpha)}{p} \geq \frac{\gamma_2}{3} >\delta$.
    \end{itemize}
  \end{observation}
  
   \begin{proof}
     Observe that
      \begin{align*}
          c'' -P_{b''s''}(\alpha) &= [-(\tfrac{1}{2} - \alpha)\gamma_1, (\tfrac{\alpha}{2} - \tfrac{1}{4})\gamma_2, -(\tfrac{1}{2}+\alpha) \gamma_2]\\
          c'' -P_{a''s''}(\alpha) &= [-(\tfrac{1}{2} - \alpha)\gamma_1, -(\tfrac{\alpha}{2} - \tfrac{1}{4})\gamma_2, -(\tfrac{1}{2}+\alpha) \gamma_2]
      \end{align*}
       
      It follows that for $\alpha \in [-\frac{1}{2},\frac{1}{2}]$ we have
       \begin{align*}
          \norm{c'' -P_{b''s''}(\alpha)}{p} &\geq  \max\big(\abs{(\tfrac{\alpha}{2} - \tfrac{1}{4}\big)\gamma_2}, \abs{(\tfrac{1}{2}+\alpha) \gamma_2}) &= \gamma_2 \cdot  \max(\abs{\tfrac{\alpha}{2} - \tfrac{1}{4}}, \abs{\tfrac{1}{2}+\alpha}) &= \tfrac{\gamma_2}{3}\\
          \norm{c'' -P_{a''s''}(\alpha)}{p} &\geq  \max\big(\abs{(\tfrac{\alpha}{2} - \tfrac{1}{4})\gamma_2}, \abs{(\tfrac{1}{2}+\alpha) \gamma_2}\big) &= \gamma_2 \cdot  \max(\abs{\tfrac{\alpha}{2} - \tfrac{1}{4}}, \abs{\tfrac{1}{2}+\alpha}) &= \tfrac{\gamma_2}{3}
      \end{align*}
       Again since we set $\gamma_2 $ to at least $4\delta$, we have $\frac{\gamma_2}{3} > \delta$.  
   \end{proof}  
  
  For any $y \in \tilde{A} \cup \tilde{B}$, we define $y'' = a''$ if $y = \tilde{a} \in \tilde{A}$ and $y'' = b''$ if $y = \tilde{b} \in \tilde{B}$. Then by Observation \ref{largedistanceC} we have $\norm{\tilde{c} - P_{sy}(\alpha)}{p} \geq \norm{c'' - P_{y''s''}(\alpha)}{p} > \delta$. This finishes the proof of Lemma~\ref{lastlemma}, and thus of Theorem~\ref{mainLower}.
\end{proof}

\section{Discussion of the \boldmath$\forall \forall \exists$-OV Hypothesis} \label{sec:discussion}

The \FOV hypothesis, that we introduced in this paper, is a special case of the following more general hypothesis (by setting $k=3$ and $Q_1 = Q_2 = \forall$).

\medskip \noindent
\textbf{Quantified-\boldmath$k$-OV Hypothesis:} \emph{Problem:} Fix quantifiers $Q_1, \ldots, Q_{k-1} \in \{\forall,\exists\}$. Given sets $A_1,\ldots,A_k \subseteq \{0,1\}^d$ of size $n$, determine whether $Q_1 a_1 \in A_1\colon \ldots Q_{k-1} a_{k-1} \in A_{k-1}\colon \exists a_k \in A_k$ such that $a_1,\ldots,a_k$ are orthogonal. \\ \emph{Hypothesis:} For any $k \ge 1$, any $Q_1,\ldots,Q_{k-1}$, and any $\varepsilon > 0$ the problem is not in time $\dO(n^{k-\varepsilon})$.
\medskip

These problems were studied by Gao et al.~\cite{GaoImpagliazzo17}, who showed that (even for every fixed $k$ and $Q_1,\ldots,Q_{k-1}$) the Quantified-$k$-OV hypothesis implies the 2-OV hypothesis. Unfortunately, there is no reduction known in the opposite direction. In fact, Carmosino et al.~\cite{CarmosinoPaturi16} established barriers for a reduction in the other direction, see also the discussion of the Hitting Set\footnote{The Hitting Set problem considered in~\cite{Amirhittingset16} is equivalent to $\forall \exists$-OV.} hypothesis in~\cite{Amirhittingset16}. 
Hence, we cannot base the hardness of Quantified-$k$-OV on the more standard $k$-OV hypothesis. 

It is well-known that the following Strong Exponential Time Hypothesis implies the $k$-OV hypothesis~\cite{Wil04}.

\medskip \noindent
\textbf{Strong Exponential Time Hypothesis (SETH)~\cite{ImpagliazzoPZ01}:} \emph{Problem:} Given a $q$-CNF formula $\phi$ over variables $x_1,\ldots,x_n$, determine whether there exist $x_1,\ldots,x_n$ such that $\phi$ evaluates to true. \\ \emph{Hypothesis:} For any $\varepsilon > 0$ there exists $q \ge 3$ such that the problem is not in time $O(2^{(1-\varepsilon)n})$.
\medskip

Similarly, we can pose a hypothesis for Quantified Satisfiability, that implies the Quantified-$k$-OV hypothesis (by essentially the same proof as in~\cite{Wil04}). 


\medskip \noindent
\textbf{Quantified-SETH:} \emph{Problem:} Given a $q$-CNF formula $\phi$ over variables $x_1,\ldots,x_n$, determine whether for all $x_1,\ldots,x_{\alpha(1) n}$ there exist $x_{\alpha(1) n + 1},\ldots,x_{\alpha(2) n}$ such that ... such that for all $x_{\alpha(2s) n+1},\ldots,x_{\alpha(2s+1) n}$ there exist $x_{\alpha(2s+1) n+1},\ldots,x_{n}$ such that $\phi$ evaluates to true. \\
\emph{Hypothesis:} For any $s \ge 0$, any $0 \le \alpha(1) < \ldots < \alpha(2s+1) < 1$, and any $\varepsilon > 0$ there exists $q \ge 3$ such that the problem is not in time $O(2^{(1-\varepsilon)n})$.
\medskip

Although Quantified Satisfiability is one of the fundamental problems studied in complexity theory (known to be PSPACE-complete), no algorithm violating Quantified-SETH is known. 

Hence, Quantified-SETH and the Quantified-$k$-OV hypothesis are two hypotheses that are even stronger than the \FOV hypothesis that we used in this paper to prove a conditional lower bound. The fact that even these stronger hypotheses have not been falsified in decades of studying these problems, we view as evidence that the \FOV hypothesis is a plausible conjecture.

\medskip


\end{document}